\documentclass[a4paper,10pt]{llncs}

\usepackage{microtype}

\usepackage{wrapfig}
\usepackage{microtype}
\usepackage{graphicx}
\usepackage{graphics}
\usepackage{epsfig}
\usepackage{epstopdf}
\usepackage{amsmath}
\usepackage{latexsym}
\usepackage{wrapfig}
\usepackage{multirow}
\usepackage{amsfonts}
\usepackage{cite}
\usepackage{amssymb}
\usepackage{caption}
\usepackage{algorithm}
\usepackage{algorithmic}
\usepackage{enumerate}
\usepackage{diagbox}
\usepackage{subfig}
\usepackage{wrapfig}
\usepackage{color}
\usepackage{hyperref}

\usepackage[a4paper,margin=2.5cm,footskip=0.25in]{geometry}

\makeatletter
\newcommand{\rmnum}[1]{\romannumeral #1}
\newcommand{\Rmnum}[1]{\expandafter\@slowromancap\romannumeral #1@}
\makeatother

\spnewtheorem{claim}{Claim}{\bfseries}{\rmfamily}
\begin{document}

\title{The Effectiveness of Uniform Sampling for Center-Based Clustering with Outliers}
\author{Hu Ding, Jiawei Huang, and Haikuo Yu}
\institute{
 School of Computer Science and Engineering, University of Science and Technology of China \\
 He Fei, China\\
  \email{huding@ustc.edu.cn}, 
    \email{hjw0330@mail.ustc.edu.cn},  \email{yhk7786@mail.ustc.edu.cn}\\
}
%
\maketitle

\thispagestyle{empty}

\begin{abstract}
Clustering has many important applications in computer science, but real-world datasets often contain outliers. Moreover, the presence of outliers can make the clustering problems to be much more challenging. To reduce the complexities, various sampling methods have been proposed in past years. Namely, we take a small sample (uniformly or non-uniformly) from input and run an existing approximation algorithm on the sample. Comparing with existing non-uniform sampling methods, the uniform sampling approach has several significant benefits. For example, it only needs to read the data in one-pass and is very easy to implement in practice. Thus, the effectiveness of uniform sampling for clustering with outliers is a natural and fundamental problem deserving to study in both theory and practice. The previous analyses on uniform sampling often indicate that the sample size should depend on the ratio $n/z$, where $n$ is the number of input points and $z$ is the pre-specified number of outliers, and the dimensionality (for instance in Euclidean space), which could be both very high. Moreover, to guarantee the desired clustering  qualities, they need to discard more than $z$ outliers. In this paper, we propose a new and unified framework for analyzing the effectiveness of uniform sampling for three representative center-based clustering with outliers problems, $k$-center/median/means clustering with outliers. We introduce a ``significance'' criterion and prove that the performance of uniform sampling depends on the significance degree of the given instance. In particular, we show that  the sample size can be independent of the ratio $n/z$ and the dimensionality. More importantly, to the best of our knowledge, our method is the first uniform sampling approach that allows to discard exactly $z$ outliers for these three center-based clustering with outliers problems. The results proposed in this paper also can be viewed as an extension of the previous sub-linear time algorithms for the ordinary clustering problems (without outliers). The experiments suggest that the uniform sampling method can achieve comparable clustering results with other existing methods, but greatly reduce the running times.

\end{abstract}

\newpage

\pagestyle{plain}
\pagenumbering{arabic}
\setcounter{page}{1}

\section{Introduction}
\label{sec-intro}
\vspace{-0.1in}

Clustering is a fundamental topic that has many important applications in real world~\cite{jain2010data}. An important type of clustering problems is called ``center-based clustering'' including the well-known $k$-center/median/means clustering problems~\cite{awasthi2014center}. Center-based clustering problems can be defined in arbitrary metrics and Euclidean space $\mathbb{R}^D$. 
Usually, a center-based clustering problem aims to find $k$ cluster centers so as to minimize the induced clustering cost. For example, the $k$-center clustering problem is to minimize the maximum distance from the input data to the set of cluster centers~\cite{hochbaum1985best,G85}; the $k$-median (means) clustering problem is to minimize the average (squared) distance instead~\cite{lloyd1982least,li2016approximating}. 
%
%

Real-world datasets often contain outliers that could seriously destroy the final clustering results~\cite{tan2006introduction,chandola2009anomaly}. Clustering with outliers can be viewed as a generalization of the ordinary clustering problems; however, the presence  of outliers makes the problems to be much more challenging. Charikar {\em et al.}~\cite{charikar2001algorithms} proposed a $3$-approximation algorithm for $k$-center clustering with outliers in arbitrary metrics. 
The time complexity of their algorithm is at least quadratic in data size, since it needs to read all the pairwise distances. A following streaming $(4+\epsilon)$-approximation algorithm was proposed by McCutchen and Khuller~\cite{mccutchen2008streaming}. 
Chakrabarty {\em et al.}~\cite{DBLP:conf/icalp/ChakrabartyGK16} showed a $2$-approximation algorithm for metric $k$-center clustering with outliers based on the LP relaxation techniques. 
Recently, Ding {\em et al.}~\cite{DBLP:conf/esa/DingYW19} provided a greedy algorithm that yields a bi-criteria approximation (returning more than $k$ clusters) based on the idea of the Gonzalez's $k$-center clustering algorithm~\cite{G85}. B\u{a}doiu {\em et al.}~\cite{BHI} showed a coreset based approach but having an exponential time complexity if $k$ is not a constant (a ``coreset'' is a small set of points that approximates the structure/shape of a much larger point set, and thus can be used to significantly reduce the time complexities for many optimization problems~\cite{DBLP:journals/widm/Feldman20}). The coresets for instances in doubling metrics were studied in~\cite{DBLP:conf/esa/DingYW19,DBLP:journals/corr/abs-1802-09205,huang2018epsilon}. 


For $k$-median/means clustering with outliers, the algorithms with provable guarantees~\cite{chen2008constant,krishnaswamy2018constant,friggstad2018approximation} are difficult to implement due to their high complexities. Several heuristic algorithms without provable guarantee have been studied before~\cite{chawla2013k,ott2014integrated}. By using the local search method, Gupta {\em et al.}\cite{gupta2017local} provided a $274$-approximation algorithm for $k$-means clustering with outliers; they also showed that the well known $k$-means++ method~\cite{arthur2007k} can be used as a coreset approach to reduce the complexity. Very recently, Im {\em et al.}\cite{DBLP:journals/corr/abs-2003-02433} provided a method for constructing the coreset of $k$-means clustering with outliers by combining $k$-means++ and uniform sampling. 
Partly inspired by the successive sampling method of~\cite{mettu2004optimal}, Chen {\em et al.}~\cite{DBLP:conf/nips/ChenA018} proposed a novel summary construction algorithm to reduce input data size. 

Moreover, due to the rapid increase of data volumes in real world, a number of communication efficient distributed algorithms for $k$-center clustering with outliers~\cite{malkomes2015fast,guha2017distributed,DBLP:journals/corr/abs-1802-09205,li2018distributed} and $k$-median/means clustering with outliers~\cite{guha2017distributed,li2018distributed,DBLP:conf/nips/ChenA018} were proposed in recent years. 

%
%
%
 \vspace{-0.05in}
\subsection{Existing Sampling Methods and Our Main Results}
\label{sec-ourresult}
 \vspace{-0.05in}

As mentioned in above, existing algorithms for clustering with outliers often have high complexities ({\em e.g.}, quadratic complexity). Therefore, several sampling methods have been studied for reducing the complexities. Namely, we take a small sample (uniformly or non-uniformly) from input and run an existing approximation algorithm on the sample. The \textbf{non-uniform sampling methods} include the aforementioned greedy algorithm~\cite{DBLP:conf/esa/DingYW19}, $k$-means++~\cite{gupta2017local,DBLP:journals/corr/abs-2003-02433}, and successive sampling~\cite{DBLP:conf/nips/ChenA018}. However, these approaches suffer several drawbacks in practice; for example, they need to read the input dataset in multiple passes with high computational complexities, or have to discard more than the pre-specified number of outliers. The sensitivity-based coreset method is also a popular non-uniform sampling approach for ordinary clustering problems~\cite{DBLP:conf/stoc/FeldmanL11}. Informally, each data point $p$ has the ``sensitivity'' $\phi(p)$ to measure its importance to the whole dataset; the coreset construction is a simple sampling procedure where each point  is drawn {\em i.i.d.}  proportional to its sensitivity. However, to the best of our knowledge, the sensitivity-based coreset approach is not quite ideal to handle outliers, as it is not easy to compute the sensitivities  because each point $p$ could be inlier or outlier for different solutions.

Due to the simplicity, the idea of \textbf{uniform sampling} has attracted a lot of attention. We follow the usual definition of ``uniform sampling'' in the articles~\cite{meyerson2004k,mishra2001sublinear,czumaj2004sublinear}, where it means that \textbf{we take a sample from the input independently and uniformly at random}. Suppose $n$ is the input size and $z$ is the number of outliers. Charikar {\em et al.}\cite{charikar2003better} and Meyerson {\em et al.}~\cite{meyerson2004k} respectively provided uniform sampling approaches for reducing data size for clustering with outliers; Huang {\em et al.}~\cite{huang2018epsilon} and Ding {\em et al.}~\cite{DBLP:conf/esa/DingYW19} presented similar results for instance in Euclidean space. However, these methods usually suffer the following dilemma.

\vspace{0.05in}
\textbf{The dilemma for uniform sampling.} Let $S$ be the uniform sample from the input. If we try to avoid sampling any outlier, the size of $S$ should not be too large; but a small $S$ is more likely to yield a large clustering error. If we keep $|S|$ large enough, it is necessary to have an accurate enough estimation on the number of sampled outliers in $S$; that means $|S|$ should be larger than some threshold depending on $n/z$ (and the dimensionality $D$ for instance in Euclidean space), which could be very large ({\em e.g.,} $D$ could be high and $z$ could be much smaller than $n$).  As the results proposed in~\cite{huang2018epsilon,DBLP:conf/esa/DingYW19}, $S$ should be at least an $\epsilon$-sample (or some other variants) with the size depending on the VC-dimension in Euclidean space and the ratio $n/z$. 
Moreover, since it is impossible to know the exact number of sampled outliers in $S$ (even if $|S|$ is large), the error on the number of discarded outliers seems to be inevitable in almost all the previous uniform sampling approaches~\cite{charikar2003better,meyerson2004k,huang2018epsilon,DBLP:conf/esa/DingYW19}; that is, to guarantee the desired clustering  qualities, they need to discard more than $z$ outliers. 
\vspace{0.05in}


The uniform sampling method was also applied to design sub-linear time algorithms for ordinary $k$-median/means clustering (without outliers) problems~\cite{meyerson2004k,mishra2001sublinear,czumaj2004sublinear,DBLP:conf/stoc/Indyk99}, and we notice that the sample sizes proposed in~\cite{meyerson2004k,czumaj2004sublinear} are independent of $n/z$ and $D$. Therefore, a natural question is that whether their results can be extended for the clustering with outliers problems; in other words, is it possible to remove the dependencies on the values $n/z$ and $D$ in the sample complexity of uniform sampling? Another key question is that whether we can discard exactly $z$ outliers when using uniform sampling. 



\textbf{Our contributions.} Though the uniform sampling approach suffers the above dilemma in theory, it often achieves nice performance in practice even if the sample size is much smaller than the theoretical bounds proposed in~\cite{charikar2003better,meyerson2004k,huang2018epsilon,DBLP:conf/esa/DingYW19}. \textbf{To explain this phenomenon, we propose a new and unified framework for analyzing the effectiveness of uniform sampling for $k$-center/median/means clustering with outliers}. 
%
We show that the sample size can be independent of the ratio $n/z$ and the dimensionality $D$, under some reasonable assumption. If we only require to output cluster centers, our uniform sampling approach runs in sub-linear time that is independent of the input size\footnote{Obviously, if we require to output the clustering membership for each point, it needs at least linear time.}. To further boost the success probability, we can take multiple samples and select the one yielding the smallest objective value by scanning the whole dataset in one-pass. More importantly, to the best of our knowledge, our method is the first uniform sampling approach that allows to discard exactly $z$ outliers.

%
%

In our framework, we consider the relation between two  values: the lower bound of the sizes of the optimal clusters $\inf_{1\leq j\leq k}|C^*_j|$ and the number of outliers $z$ (the formal definitions will be shown in Section~\ref{sec-pre}). 
 If a cluster $C^*_j$ has size $\ll z$, then we can say that $C^*_j$ is not a ``significant'' cluster. In real applications, we may only have an estimation for the number of clusters $k$. Consequently, if there exists a cluster $C^*_j$ having size $|C^*_j|\ll z$, we can  formulate the problem as a simpler $(k-1)$-center/median/means clustering with $z+|C^*_j|$ outliers instead. So we can assume that each cluster $C^*_j$ has a size comparable to $z$. If $\frac{\inf_{1\leq j\leq k}|C^*_j|}{z}=\Omega(1)$ (note that the ratio could be a value smaller than $1$, say $0.5$, in this case), our framework outputs $k+O(\log k)$ cluster centers that yield a $4$-approximation for $k$-center clustering with outliers; further, if $\inf_{1\leq j\leq k}|C^*_j|>z$, our framework returns exactly $k$ cluster centers that yield a $(c+2)$-approximation solution, if we run an existing $c$-approximation algorithm with $c\geq 1$ on the sample. The framework can also handle $k$-median/means clustering with outliers and yields similar results.

 
We should point out that Meyerson {\em et al.}~\cite{meyerson2004k} also considered the lower bound of the cluster sizes when designing their sub-linear time $k$-median clustering (without outliers) algorithm. However, it is challenging to directly extend their idea to handle the case  with outliers (as explained in the aforementioned dilemma for uniform sampling), and thus we need to develop significantly new ideas in our algorithms design and analysis. 
Recently, Gupta~\cite{gupta2018approximation} proposed a similar uniform sampling approach to handle $k$-means clustering with outliers. 
However, their analysis and results are quite different from ours. Also their assumption is stronger: it requires that each optimal cluster has size roughly $\Omega(\frac{z}{\gamma^2}\log k)$ ({\em i.e.,} $\frac{\inf_{1\leq j\leq k}|C^*_j|}{z}=\Omega(\frac{\log k}{\gamma^2})$) where $\gamma$ is a small parameter in $(0,1)$.

 \vspace{-0.05in}
\subsection{Preliminaries}
\label{sec-pre}
\vspace{-0.05in}
Let the input be a point set $P\subset \mathbb{R}^D$ with $|P|=n$. Given a set of points $H\subset \mathbb{R}^D$ and a positive integer $z<n$, we define the following notations.

\begin{eqnarray}
\Delta^{-z}_{\infty}(P, H)&=&\min\{\max_{p\in P'}dist(p, H)\mid P'\subset P, |P'|=n-z\};\label{def-center}\\
%
\Delta^{-z}_{1}(P, H)&=&\min\{\frac{1}{|P'|}\sum_{p\in P'}dist(p, H)\mid P'\subset P, |P'|=n-z\};\label{def-median}\\
\Delta^{-z}_{2}(P, H)&=&\min\{\frac{1}{|P'|}\sum_{p\in P'}\big(dist(p, H)\big)^2\mid P'\subset P, |P'|=n-z\},\label{def-means}
\end{eqnarray}
where $dist(p, H)=\min_{q\in H}||p-q||$ and $||p-q||$ denotes the distance between $p$ and $q$.

\begin{definition}[$k$-Center/Median/Means Clustering with Outliers]
\label{def-outlier}
Given a set $P$ of $n$ points in $\mathbb{R}^D$ with two positive integers $k$ and $z<n$, the problem of $k$-center (resp., $k$-median, $k$-means) clustering with outliers is to find $k$ cluster centers $C=\{c_1, \cdots, c_k\}\subset \mathbb{R}^D$, such that the objective function $\Delta^{-z}_{\infty}(P, C)$ (resp., $\Delta^{-z}_{1}(P, C)$, $\Delta^{-z}_{2}(P, C)$) is minimized.
%
\end{definition}
The definition can be easily modified for arbitrary metric space $(X, d)$, where $X$ contains $n$ vertices and $d(\cdot, \cdot)$ is the distance function: the Euclidean distance ``$||p-q||$'' is replaced by $d(p, q)$; the cluster centers $\{c_1, \cdots, c_k\}$ should be chosen from $X$. In this paper, we always use $P_{opt}$, a subset of $P$ with size $n-z$, to denote the subset yielding the optimal solution with respect to the objective functions in Definition~\ref{def-outlier}. Also, let $\{C^*_1, \cdots, C^*_k\}$ be the $k$ optimal clusters forming $P_{opt}$. 


As mentioned in Section~\ref{sec-ourresult}, it is rational to assume that $\inf_{1\leq j\leq k} |C^*_j|$ is not far smaller than $z$. To formally state this assumption, we introduce the following definition. 

\begin{definition}[$(\epsilon_1, \epsilon_2)$-Significant Instance]
\label{def-sig}
Let $\epsilon_1, \epsilon_2>0$. Given an instance of $k$-center (resp., $k$-median, $k$-means) clustering with outliers as described in Definition~\ref{def-outlier}, if $\inf_{1\leq j\leq k} |C^*_j|\geq\frac{\epsilon_1}{k}n$ and $z=\frac{\epsilon_2}{k}n$, we say that it is an $(\epsilon_1,\epsilon_2)$-significant instance.
\end{definition}
Obviously, since $\sum^k_{j=1}|C^*_j|<n$, $\epsilon_1$ should be smaller than $1$. In Definition~\ref{def-sig}, we do not say ``$\inf_{1\leq j\leq k} |C^*_j|=\frac{\epsilon_1}{k}n$'', since we may not be able to obtain the exact value of $\inf_{1\leq j\leq k} |C^*_j|$; instead, we may only have a lower bound of $\inf_{1\leq j\leq k} |C^*_j|$. The ratio $\frac{\epsilon_1}{\epsilon_2}$ ($\leq\frac{\inf_{1\leq j\leq k} |C^*_j|}{z}$) reveals the ``significance'' of the clusters to outliers; the higher the ratio is, the more significant the clusters to outliers will be. 


The rest of the paper is organized as follows. To help our analysis, we present two implications of Definition~\ref{def-sig} in Section~\ref{sec-implication}. Then we introduce our uniform sampling algorithms for $k$-center clustering with outliers and $k$-median/means clustering with outliers in Section~\ref{sec-kcenter} and~\ref{sec-kmedian}, respectively. We also explain that how to boost the success probability of our framework and how to further determine the clustering memberships of data points in Section~\ref{sec-boost}. Finally, we present our experimental results in Section~\ref{sec-exp}.

\vspace{-0.1in}
\section{Implications of Significant Instance}
\label{sec-implication}
\vspace{-0.1in}

\begin{lemma}
\label{lem-imp1}
Given an $(\epsilon_1, \epsilon_2)$-significant instance $P$ as described in Definition~\ref{def-sig}, one selects a set $S$ of points from $P$ uniformly at random. Let $\eta, \delta\in(0,1)$. (\rmnum{1}) If $|S|\geq \frac{k}{\epsilon_1}\log\frac{k}{\eta}$, with probability at least $1-\eta$, $S\cap C^*_j\neq \emptyset$ for any $1\leq j\leq k$. (\rmnum{2}) If $|S|\geq\frac{3k}{\delta^2\epsilon_1}\log\frac{2k}{\eta}$, with probability at least $1-\eta$, $|S\cap C^*_j|\in (1\pm\delta)\frac{|C^*_j|}{n} |S|$ for any $1\leq j\leq k$.
\end{lemma}


Lemma~\ref{lem-imp1} can be obtained by using the Chernoff bound~\cite{alon2004probabilistic}, and we leave the proof to Section~\ref{sec-proof-lem-imp1}.
Moreover, we know that the expected number of outliers contained in the sample $S$ is $\frac{\epsilon_2}{k}|S|$. So we immediately have the following result by using the Markov's inequality.

\begin{lemma}
\label{lem-imp2}
Given an $(\epsilon_1, \epsilon_2)$-significant instance $P$ as described in Definition~\ref{def-sig}, one selects a set $S$ of points from $P$ uniformly at random. Let $\eta\in(0,1)$. With probability at least $1-\eta$, $\big|S\setminus P_{opt}\big|\leq  \frac{\epsilon_2}{k\eta}|S|$.
\end{lemma}
\vspace{-0.1in}

\section{Uniform Sampling for $k$-Center Clustering with Outliers}
\label{sec-kcenter}
\vspace{-0.1in}

Let $r_{opt}$ be the optimal radius of the instance $P$, {\em i.e.}, each optimal cluster $C^*_j$ is covered by a ball with radius $r_{opt}$. For any point $p\in \mathbb{R}^D$ and any value $r\geq 0$, we use $Ball(p, r)$ to denote the ball centered at $p$ with radius $r$. We present  Algorithm~\ref{alg-kc1} and \ref{alg-kc2} for the $k$-center clustering with outliers problem, and prove their clustering qualities in Theorem~\ref{the-kc1} and \ref{theorem-kc2} respectively. We only focus on the problem in Euclidean space due to the space limit, but the results also hold for abstract metric space by using the same idea.

\begin{theorem}
\label{the-kc1}
In Algorithm~\ref{alg-kc1}, the size $|H|= k+\frac{1}{\eta}\frac{\epsilon_2}{\epsilon_1}\log \frac{k}{\eta}$. Also, with probability at least $(1-\eta)^2$, $\Delta^{-z}_{\infty}(P, H)\leq 4r_{opt}$. 
\end{theorem}
\begin{remark}
\label{rem-the-kc1}
\textbf{(\rmnum{1})} If we assume both $\frac{\epsilon_2}{\epsilon_1}$ and $\frac{1}{\eta}$ are $O(1)$, Theorem~\ref{the-kc1} indicates that Algorithm~\ref{alg-kc1} returns $k+O(\log k)$ cluster centers. $\frac{\epsilon_2}{\epsilon_1}=O(1)$ means that $\frac{\inf_{1\leq j\leq k}|C^*_j|}{z}=\Omega(1)$, that is, the number of outliers is not significantly larger than the size of the smallest cluster. The running time of the subroutine algorithm~\cite{G85} in Step 2 is $(k+k')|S|D=O(\frac{k^2}{\epsilon_1}(\log k) D)$, which is independent of the input size $n$. 
\textbf{(\rmnum{2})} In Section~\ref{sec-example}, we present an example to show that the value of $k'$ cannot be further reduced. That is, the clustering quality could be arbitrarily bad if we run $(k+k'')$-center clustering on $S$ with $k''<k'$. 
\end{remark}
\vspace{-0.2in}
\begin{proof}(\textbf{of Theorem~\ref{the-kc1}})
First,  it is straightforward to know that $|H|=k+k'= k+\frac{1}{\eta}\frac{\epsilon_2}{\epsilon_1}\log \frac{k}{\eta}$. 
Below, we assume that the sample $S$ contains at least one point from each $C^*_j$, and at most $k'=\frac{\epsilon_2}{k\eta}|S|$ points from $P\setminus P_{opt}$ (these events happen with probability at least $(1-\eta)^2$ due to Lemma~\ref{lem-imp1} and \ref{lem-imp2}).

Since the sample $S$ contains at most $k'$ points from $P\setminus P_{opt}$ and $P_{opt}$ can be covered by $k$ balls with radius $r_{opt}$, we know that $S$ can be covered by $k+k'$ balls with radius $r_{opt}$. Thus, if we perform the $2$-approximation $(k+k')$-center clustering algorithm~\cite{G85} on $S$, the resulting balls should have radius no larger than $2r_{opt}$. Let $H=\{h_1, \cdots, h_{k+k'}\}$ and $\mathbb{B}_S=\{Ball(h_l, r)\mid 1\leq l\leq k+k'\}$ be those balls covering $S$ with $r\leq 2r_{opt}$. Also, for each $1\leq j\leq k$, since $S\cap C^*_j\neq\emptyset$, there exists one ball of $\mathbb{B}_S$, say $Ball(h_{l_j}, r)$, covers at least one point, say $p_j$, from $C^*_j$. For any point $p\in C^*_j$, we have $||p-p_j||\leq 2 r_{opt}$ (by the triangle inequality) and $||p_j-h_{l_j}||\leq r\leq 2 r_{opt}$; therefore, $||p-h_{l_j}||\leq ||p-p_j||+||p_j-h_{l_j}||\leq 4 r_{opt}$. See Figure~\ref{fig-th1}. Overall, $P_{opt}=\cup^k_{j=1}C^*_j$ is covered by the balls $\cup^{k+k'}_{l=1}Ball(h_l, 4 r_{opt})$, {\em i.e.}, $\Delta^{-z}_{\infty}(P, H)\leq 4r_{opt}$. 
%
%
\qed
\end{proof}

\begin{algorithm}[tb]
   \caption{\textsc{Uniform Sampling $k$-Center Outliers \Rmnum{1}}}
   \label{alg-kc1}
\begin{algorithmic}
  \STATE {\bfseries Input:} An $(\epsilon_1, \epsilon_2)$-significant instance $P$ of $k$-center clustering with $z$ outliers, and $|P|=n$; a parameter $\eta\in (0,1)$. 
   \STATE
\begin{enumerate}
\item Sample a set $S$ of $\frac{k}{\epsilon_1}\log\frac{k}{\eta}$ points uniformly at random from $P$.
\item Let $k'=\frac{1}{\eta}\frac{\epsilon_2}{k}|S|$, and solve the $(k+k')$-center clustering problem on $S$ by using the $2$-approximation algorithm~\cite{G85}. 
\end{enumerate}
  \STATE {\bfseries Output}  $H$, which is the set of $k+k'$ cluster centers returned in Step 2. 
\end{algorithmic}
\end{algorithm}

\begin{figure}
 \vspace{-0.2in}
\begin{center}
    \includegraphics[width=0.17\textwidth]{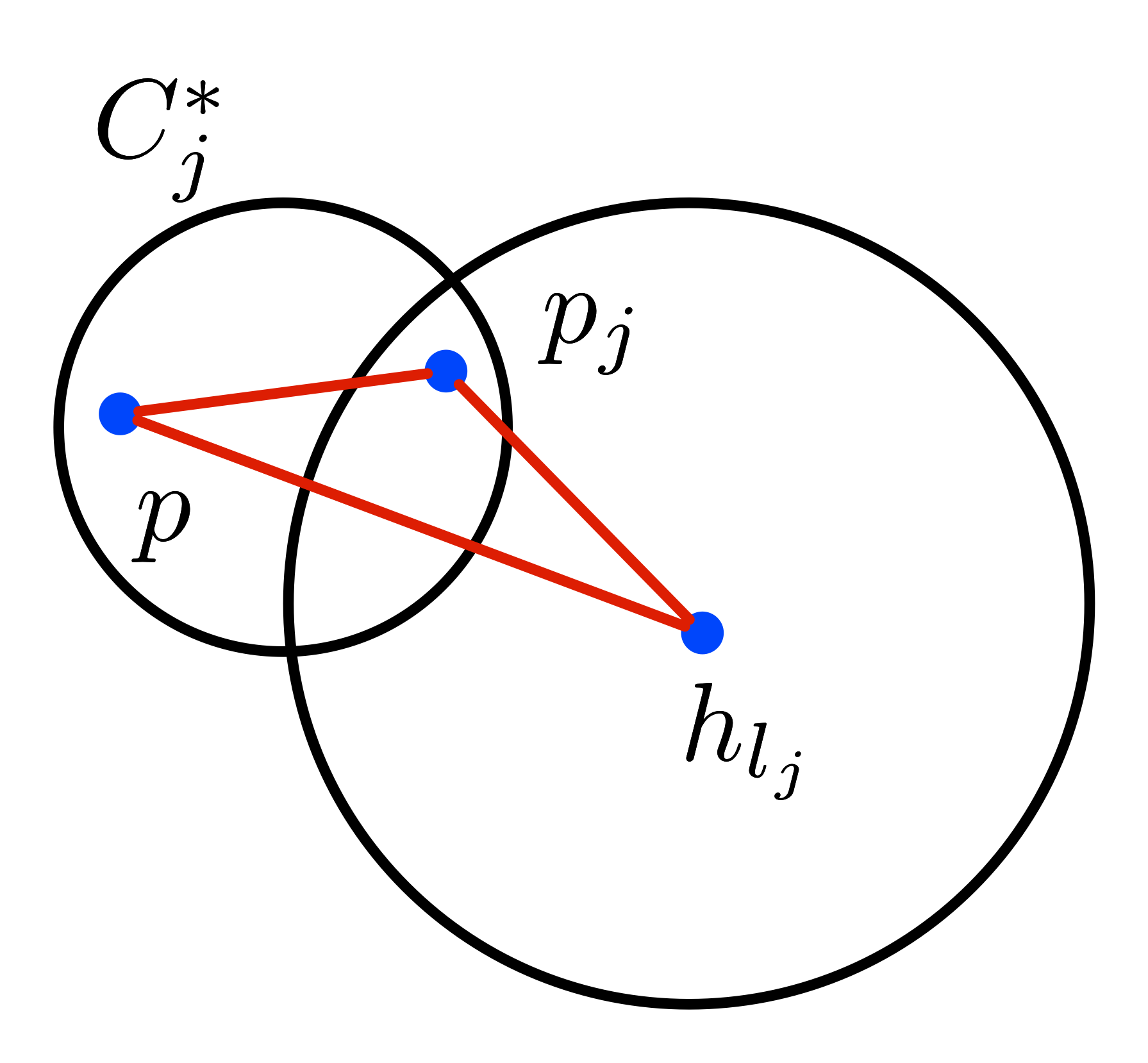}  
    \end{center}
 \vspace{-0.2in}
  \caption{$||p-h_{l_j}||\leq ||p-p_j||+||p_j-h_{l_j}||\leq 4 r_{opt}$.}     
   \label{fig-th1}
 \vspace{-0.3in}
\end{figure}

\begin{algorithm}[tb]
   \caption{\textsc{Uniform Sampling $k$-Center Outliers \Rmnum{2}}}
   \label{alg-kc2}
\begin{algorithmic}
  \STATE {\bfseries Input:}  An $(\epsilon_1, \epsilon_2)$-significant instance $P$ of $k$-center clustering with $z$ outliers, and $|P|=n$; two parameters $\eta,\delta\in (0,1)$. 
   \STATE
\begin{enumerate}
\item Sample a set $S$ of $\frac{3k}{\delta^2\epsilon_1}\log\frac{2k}{\eta}$ points uniformly at random from $P$.
\item Let $z'=\frac{1}{\eta}\frac{\epsilon_2}{k}|S|$, and solve the $k$-center clustering with $z'$ outliers problem on $S$ by using any $c$-approximation algorithm with $c\geq1$ ({\em e.g.,} the $3$-approximation algorithm~\cite{charikar2001algorithms}). 
\end{enumerate}
  \STATE {\bfseries Output} $H$, which is the set of $k$ cluster centers returned in Step 2.
\end{algorithmic}
\end{algorithm}


\begin{theorem}
\label{theorem-kc2}
If $\frac{\epsilon_1}{\epsilon_2}>\frac{1}{\eta(1-\delta)}$, with probability at least $(1-\eta)^2$, Algorithm~\ref{alg-kc2} returns $k$ cluster centers achieving a $(c+2)$-approximation for  $k$-center clustering with $z$ outliers, {\em i.e.}, $\Delta^{-z}_{\infty}(P, H)\leq (c+2)r_{opt}$.
\end{theorem}
\begin{remark}
\label{rem-2}
For example, if we set $\eta=\delta=1/2$, the algorithm works for any instance with $\frac{\epsilon_1}{\epsilon_2}>\frac{1}{\eta(1-\delta)}=4$. Actually, as long as $\frac{\epsilon_1}{\epsilon_2}>1$ ({\em i.e.,} $\inf_{1\leq j\leq k}|C^*_j|>z$), we can always find the appropriate values for $\eta$ and $\delta$ to satisfy $\frac{\epsilon_1}{\epsilon_2}>\frac{1}{\eta(1-\delta)}$. For example, we can set $\eta>\sqrt{\frac{\epsilon_2}{\epsilon_1}}$ and $\delta<1-\sqrt{\frac{\epsilon_2}{\epsilon_1}}$; obviously, if $\frac{\epsilon_1}{\epsilon_2}$ is close to $1$, the success probability $(1-\eta)^2$ could be small (we will show that how to boost the success probability in Section~\ref{sec-boost}). The running time  depends on the complexity of the subroutine $c$-approximation algorithm used in Step 2. For example, the algorithm of~\cite{charikar2001algorithms} takes $O\big(|S|^2D+k|S|^2\log |S|\big)$ time in $\mathbb{R}^D$.
\end{remark}

\begin{proof}(\textbf{of Theorem~\ref{theorem-kc2}})
We assume that $|S\cap C^*_j|\in(1\pm\delta)\frac{|C^*_j|}{n} |S|$ for each $C^*_j$, and $S$ has at most $z'=\frac{\epsilon_2}{k\eta}|S|$ points from $P\setminus P_{opt}$ (these events happen with probability at least $(1-\eta)^2$ due to Lemma~\ref{lem-imp1} and \ref{lem-imp2}).


Let $\mathbb{B}_S=\{Ball(h_l, r)\mid 1\leq l\leq k\}$ be the set of $k$ balls returned in Step 2 of Algorithm~\ref{alg-kc2}. Since $S\cap P_{opt}$ can be covered by $k$ balls with radius $r_{opt}$ and $|S\setminus P_{opt}|\leq z'$, the optimal radius for the instance $S$ with $z'$ outliers should be at most $r_{opt}$. 
Consequently, $r\leq c r_{opt}$. Moreover,
\begin{eqnarray}
|S\cap C^*_j|&\geq& (1-\delta)\frac{|C^*_j|}{n} |S|\geq (1-\delta)\frac{\epsilon_1}{k}|S|\nonumber\\
&>&(1-\delta)\frac{\epsilon_2}{\eta(1-\delta)k}|S|=\frac{\epsilon_2}{\eta k}|S|=z'\label{for-kc2-1}
\end{eqnarray}
for any $1\leq j\leq k$, where the last inequality comes from $\frac{\epsilon_1}{\epsilon_2}>\frac{1}{\eta(1-\delta)}$. Thus, if we perform $k$-center clustering with $z'$ outliers on $S$, the resulting $k$ balls must cover at least one point from each $C^*_j$ (since $|S\cap C^*_j|>z'$ from (\ref{for-kc2-1})). Through a similar manner as the proof of Theorem~\ref{the-kc1}, we know that $P_{opt}=\cup^k_{j=1}C^*_j$ is covered by the balls $\cup^{k}_{l=1}Ball(h_l, r+2 r_{opt})$, {\em i.e.}, $\Delta^{-z}_{\infty}(P, H)\leq r+2 r_{opt}\leq (c+2)r_{opt}$. 
%
%
%
%
\qed
\end{proof}

\section{Uniform Sampling for $k$-Median/Means Clustering with Outliers}
\label{sec-kmedian}
For the problem of $k$-means clustering with outliers, 
we apply the similar  ideas as Algorithm~\ref{alg-kc1} and~\ref{alg-kc2} (see Algorithm~\ref{alg-km} and~\ref{alg-km2}). However, the analyses are more complicated here. For ease of understanding, we present our high-level idea first. Also, due to the space limit, we show the extensions for  $k$-median clustering with outliers and their counterparts in arbitrary metric space in Section~\ref{sec-extension}.

\textbf{High-level idea.} Let $S$ be a large enough random sample from $P$.
 Denote by $O^*=\{o^*_1, \cdots, o^*_k\}$ the mean points of $\{C^*_1, \cdots, C^*_k\}$, respectively. We first show that $S\cap C^*_j$ can well approximate $C^*_j$ for each $1\leq j\leq k$. Informally, 
 \begin{eqnarray}
 \frac{|S\cap C^*_j|}{|S|}\approx \frac{|C^*_j|}{n} \text{ \hspace{0.1in }and \hspace{0.1in }} \frac{1}{|S\cap C^*_j|}\sum_{q\in S\cap C^*_j}||q-o^*_j||^2\approx\frac{1}{|C^*_j|}\sum_{p\in  C^*_j}||p-o^*_j||^2. \label{for-key}
 \end{eqnarray} 
We also define a transformation on $P_{opt}$ to help our analysis. 


\begin{definition}[Star Shaped Transformation]
\label{def-star}
For each point in $C^*_j$, we translate it to $o^*_j$; overall, we generate a new set of $n-z$ points located at $\{o^*_1, \cdots o^*_k\}$, where each $o^*_j$ has $|C^*_j|$ overlapping points. For any point $p\in C^*_j$ with $1\leq j\leq k$, denote by $\tilde{p}$ its transformed point; for any $U\subseteq P_{opt}$, denote by $\tilde{U}$ its transformed point set. 
\end{definition}

 \begin{wrapfigure}{r}{0.3\textwidth}
  \begin{center}
   \vspace{-0.45in}
    \includegraphics[width=0.22\textwidth]{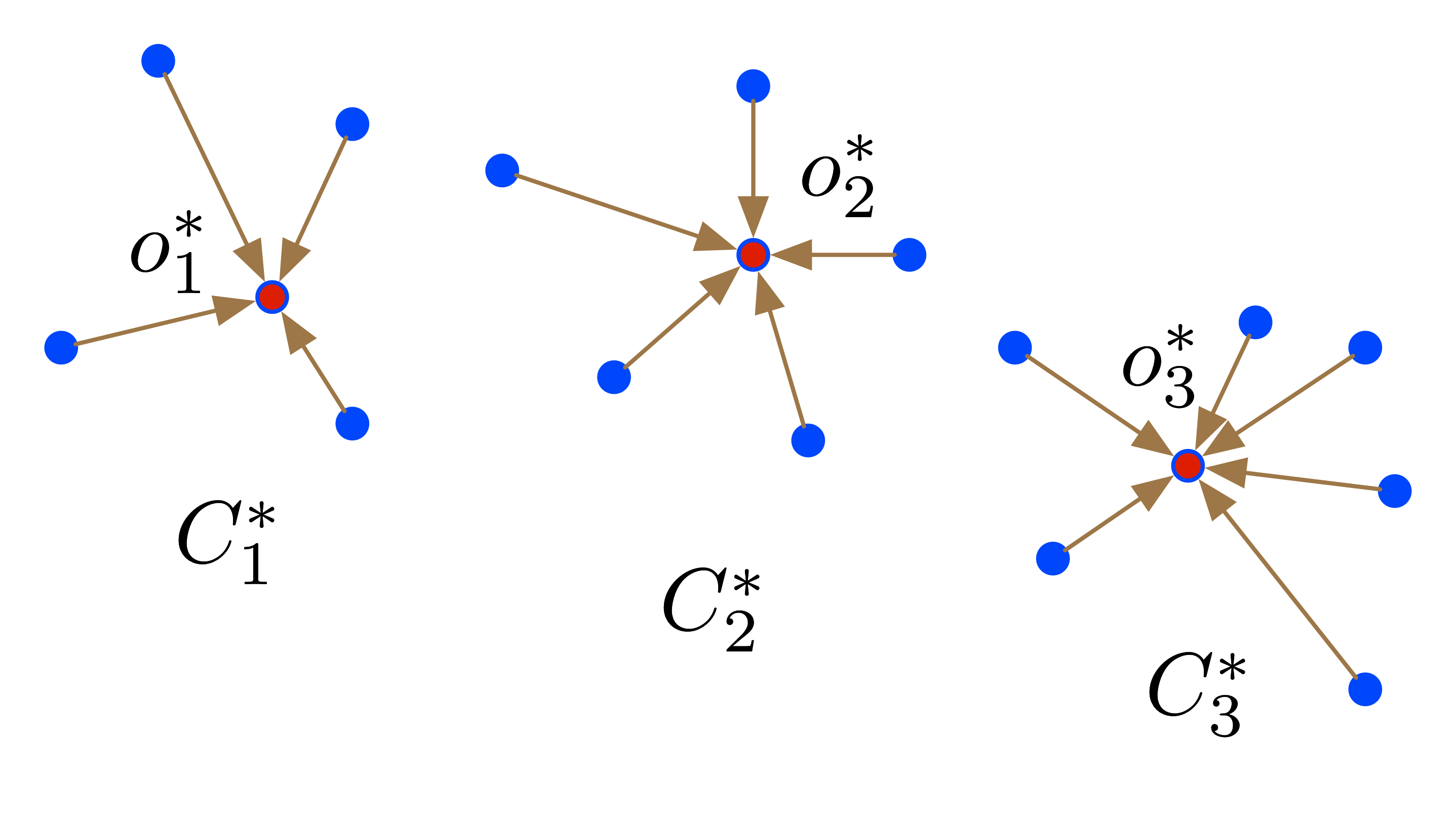}
  \end{center}
   \vspace{-0.3in}
  \caption{The transformation from $P_{opt}$ to $\tilde{P}_{opt}$.}
     \label{fig-key}
      \vspace{-0.4in}
\end{wrapfigure}

Since the transformation forms $k$ ``stars'' (see Figure.~\ref{fig-key}), we call it ``star shaped transformation''. 
%
Let $S_{opt}=S\cap P_{opt}$. By using (\ref{for-key}), we can prove that the clustering costs of $\tilde{P}_{opt}$ and $\tilde{S}_{opt}$ are close (after the normalization) for any given set of cluster centers. Let $H$ be the $k+k'$ cluster centers returned by Algorithm~\ref{alg-km}. 
Then, we can use $\tilde{S}_{opt}$ and $\tilde{P}_{opt}$ as the ``bridges'' between $S$ and $P$, so as to prove that $H$ yields an approximate solution for the instance $P$, {\em i.e.}, $\Delta^{-z}_2(P, H)$ is bounded.   

In Algorithm~\ref{alg-km2}, we run $k$-means with $z'$ outliers algorithm on the sample $S$ and return $k$ (rather than $k+k'$) cluster centers. So we need to modify the above idea to analyze the quality. Let $S_{in}$ be the set of $|S|-z'$ inliers of $S$ obtained in Step 2. If the ratio $\frac{\epsilon_1}{\epsilon_2}$ is large enough, we can prove that $|S_{in}\cap C^*_j|\approx |S\cap C^*_j|$ for each $1\leq j\leq k$. Therefore, we can replace ``$S$'' by ``$S_{in}$'' in (\ref{for-key}) and prove a similar quality guarantee for Algorithm~\ref{alg-km2}.

%

In Theorem~\ref{the-km} and~\ref{the-km2}, $\mathcal{L}$ denotes the maximum diameter of the $k$ clusters $C^*_1, \cdots, C^*_k$, {\em i.e.}, $\mathcal{L}$$=\max_{1\leq j\leq k}\max_{p, q\in C^*_j}$ $||p-q||$. Actually, our result can be viewed as an extension of the sub-linear time $k$-median/means clustering algorithms~\cite{meyerson2004k,mishra2001sublinear,czumaj2004sublinear} to the case with outliers. We also want to emphasize that the additive error is unavoidable even for the vanilla case (without outliers), if we require the sample complexity to be independent of the input size~\cite{mishra2001sublinear,czumaj2004sublinear}; the $k$-median clustering with outliers algorithm proposed by Meyerson {\em et al.}~\cite{meyerson2004k} does not yield an additive error, but it needs to discard more than $16z$ outliers and the sample size depends on the ratio $n/z$. We place the full proofs of Theorem~\ref{the-km} and~\ref{the-km2} in Section~\ref{sec-proof-km} and \ref{sec-proof-km2} respectively.  

\begin{theorem}
\label{the-km}
With probability at least $(1-\eta)^3$, the set of cluster centers $H$ returned by Algorithm~\ref{alg-km} results in a clustering cost $\Delta^{-z}_{2}(P, H)$ at most $\alpha \Delta^{-z}_{2}(P, O^*)+\beta\xi \mathcal{L}^2$, where $\alpha=\big(2+(4+4c)\frac{1+\delta}{1-\delta}\big)$ and $\beta=(4+4c)\frac{1+\delta}{1-\delta}$.
\end{theorem}

\begin{remark}
\label{rem-km}
\textbf{(\rmnum{1})} 
In Step 2 of Algorithm~\ref{alg-km}, we can apply an $O(1)$-approximation $k$-means algorithm ({\em e.g.,} \cite{kanungo2004local}). If we assume $1/\delta$ and $1/\eta$ are fixed constants, 
then  the sample size $|S|=O( \frac{k}{\xi^2\epsilon_1}\log k)$, and both the factors $\alpha$ and $\beta$ are $O(1)$, {\em i.e.}, 
$\Delta^{-z}_{2}(P, H)\leq O(1)\Delta^{-z}_{2}(P, O^*)+O(1)\xi \mathcal{L}^2$; moreover, the number of returned cluster centers $|H|=k+O(\frac{\log k}{\xi^2})$ if $\frac{\epsilon_1}{\epsilon_2}=\Omega(1)$. 
\textbf{(\rmnum{2})} Similar to Theorem~\ref{the-kc1}, we use the same example in Section~\ref{sec-example} to show that the value of $k'$ cannot be further reduced in Algorithm~\ref{alg-km}. 
 \end{remark}

\begin{theorem}
\label{the-km2}
Assume $t=\eta(1-\delta)\frac{\epsilon_1}{\epsilon_2}>1$.
With probability at least $(1-\eta)^3$, the set of cluster centers $H$ returned by Algorithm~\ref{alg-km2} results in a clustering cost $\Delta^{-z}_{2}(P, H)$ at most $\alpha \Delta^{-z}_{2}(P, O^*)+\beta\xi \mathcal{L}^2$, where $\alpha=\big(2+(4+4c)\frac{t}{t-1}\frac{1+\delta}{1-\delta}\big)$ and $\beta=(4+4c)\frac{t}{t-1}\frac{1+\delta}{1-\delta}$.
\end{theorem}
Similar to Theorem~\ref{theorem-kc2}, as long as $\epsilon_1/\epsilon_2>1$, we can set $\eta>\sqrt{\frac{\epsilon_2}{\epsilon_1}}$ and $\delta<1-\sqrt{\frac{\epsilon_2}{\epsilon_1}}$ to keep $t>1$. 

%

\begin{algorithm}[tb]
   \caption{\textsc{Uniform Sampling $k$-Means Outliers \Rmnum{1}}}
   \label{alg-km}
\begin{algorithmic}
  \STATE {\bfseries Input:} An $(\epsilon_1, \epsilon_2)$-significant instance $P$ of $k$-means clustering with $z$ outliers in $\mathbb{R}^D$, and $|P|=n$; three parameters $\eta, \delta, \xi\in (0,1)$. 
   \STATE
\begin{enumerate}
\item Sample a set $S$ of $\max\{\frac{3k}{\delta^2\epsilon_1}\log\frac{2k}{\eta}, \frac{k}{2\xi^2\epsilon_1(1-\delta)}\log\frac{2k}{\eta}\}$ points uniformly at random from $P$.
\item Let $k'=\frac{1}{\eta}\frac{\epsilon_2}{k}|S|$, and solve the $(k+k')$-means clustering on $S$ by using any $c$-approximation algorithm with $c\geq 1$.  
\end{enumerate}
  \STATE {\bfseries Output} $H$, which is the set of $k+k'$ cluster centers returned in Step 2. 
\end{algorithmic}
\end{algorithm}

\begin{algorithm}[tb]
   \caption{\textsc{Uniform Sampling $k$-Means Outliers \Rmnum{2}}}
   \label{alg-km2}
\begin{algorithmic}
  \STATE {\bfseries Input:} An $(\epsilon_1, \epsilon_2)$-significant instance $P$ of $k$-means clustering with $z$ outliers in $\mathbb{R}^D$, and $|P|=n$; three parameters $\eta, \delta, \xi\in (0,1)$. 
   \STATE
\begin{enumerate}
\item Sample a set $S$ of $\max\{\frac{3k}{\delta^2\epsilon_1}\log\frac{2k}{\eta}, \frac{k}{2\xi^2\epsilon_1(1-\delta)}\log\frac{2k}{\eta}\}$ points uniformly at random from $P$.
\item Let $z'=\frac{1}{\eta}\frac{\epsilon_2}{k}|S|$, and solve the $k$-means clustering with $z'$ outliers on $S$ by using any $c$-approximation algorithm with $c\geq 1$. 
\end{enumerate}
  \STATE {\bfseries Output}  $H$, which is the set of $k$ cluster centers returned in Step 2. 
\end{algorithmic}
\end{algorithm}

\subsection{Proof of Theorem~\ref{the-km}}
\label{sec-proof-km}
The following lemma can be obtained via the Hoeffding's inequality (each $||q-o^*_j||^2$ can be viewed as a random variable between $0$ and $\mathcal{L}^2$)~\cite{alon2004probabilistic}.

\begin{lemma}
\label{lem-km-sample1}
We fix a cluster $C^*_j$. Given $\eta, \xi\in (0,1)$, if one uniformly selects a set $T$ of $\frac{1}{2\xi^2}\log\frac{2}{\eta}$ or more points at random from $C^*_j$, 
\begin{eqnarray}
\bigg|\frac{1}{|T|}\sum_{q\in T}||q-o^*_j||^2-\frac{1}{|C^*_j|}\sum_{p\in C^*_j}||p-o^*_j||^2\bigg|\leq \xi \mathcal{L}^2
\end{eqnarray}
with probability at least $1-\eta$.
\end{lemma}


\begin{lemma}
\label{lem-km-sample3}
If one uniformly selects a set $S$ of $\max\{\frac{3k}{\delta^2\epsilon_1}\log\frac{2k}{\eta}, \frac{k}{2\xi^2\epsilon_1(1-\delta)}\log\frac{2k}{\eta}\}$ points at random from $P$, 
\begin{eqnarray}
\sum_{q\in S\cap C^*_j}||q-o^*_j||^2\leq (1+\delta)\frac{|S|}{n}\big(\sum_{p\in C^*_j}||p-o^*_j||^2+\xi |C^*_j| \mathcal{L}^2\big)
\end{eqnarray}
for $1\leq j\leq k$, with probability at least $(1-\eta)^2$.
\end{lemma}
\begin{proof}
Suppose $|S|=\max\{\frac{3k}{\delta^2\epsilon_1}\log\frac{2k}{\eta}, \frac{k}{2\xi^2\epsilon_1(1-\delta)}\log\frac{2k}{\eta}\}$. According to Lemma~\ref{lem-imp1}, $|S|\geq\frac{3k}{\delta^2\epsilon_1}\log\frac{2k}{\eta}$ indicates that $|S\cap C^*_j|\geq (1-\delta)\frac{|C^*_j|}{n}|S|\geq (1-\delta)\frac{\epsilon_1}{k}|S|$ for each $1\leq j\leq k$; further, $|S|\geq\frac{k}{2\xi^2\epsilon_1(1-\delta)}\log\frac{2k}{\eta}$ implies that $(1-\delta)\frac{\epsilon_1}{k}|S|\geq \frac{1}{2\xi^2}\log\frac{2k}{\eta}$. Therefore, $|S\cap C^*_j|\geq\frac{1}{2\xi^2}\log\frac{2k}{\eta}$ and by Lemma~\ref{lem-km-sample1} we have
\begin{eqnarray}
\bigg|\frac{1}{|S\cap C^*_j|}\sum_{q\in S\cap C^*_j}||q-o^*_j||^2-\frac{1}{|C^*_j|}\sum_{p\in C^*_j}||p-o^*_j||^2\bigg|\leq \xi\mathcal{L}^2 \label{for-km-sample3-1}
\end{eqnarray}
with probability at least $1-\eta$ for each $1\leq j\leq k$ ($\eta$ is replaced by $\eta/k$ in  Lemma~\ref{lem-km-sample1} for taking the union bound). From (\ref{for-km-sample3-1}) we know that
\begin{eqnarray}
\sum_{q\in S\cap C^*_j}||q-o^*_j||^2&\leq& |S\cap C^*_j|\big(\frac{1}{|C^*_j|}\sum_{p\in C^*_j}||p-o^*_j||^2+\xi \mathcal{L}^2\big)\nonumber\\
&\leq&(1+\delta)\frac{|C^*_j|}{n}|S|\big(\frac{1}{|C^*_j|}\sum_{p\in C^*_j}||p-o^*_j||^2+\xi \mathcal{L}^2\big)\nonumber\\
&=&(1+\delta)\frac{|S|}{n}\big(\sum_{p\in C^*_j}||p-o^*_j||^2+\xi |C^*_j| \mathcal{L}^2\big), 
\end{eqnarray}
where the second inequality comes from Lemma~\ref{lem-imp1}. So we complete the proof.
\qed
\end{proof}

%


%

We define a new notation that is used in the following lemmas. Given two point sets $X$ and $Y\subset\mathbb{R}^D$, we use $Cost(X, Y)$ to denote the clustering cost of $X$ by taking $Y$ as the cluster centers, {\em i.e.}, $Cost(X, Y)=\sum_{q\in X}(dist(q, Y))^2$. Obviously, $\Delta^{-z}_2(P, H)=\frac{1}{n-z}Cost(P_{opt}, H)$. Let $S_{opt}=S\cap P_{opt}$. 
Below, we prove the upper bounds of $Cost(S_{opt}, O^*)$, $Cost(\tilde{S}_{opt}, H)$, and $Cost(\tilde{P}_{opt}, H)$ respectively, and use these bounds to complete the proof of Theorem~\ref{the-km}. For convenience, we always assume that the events mentioned in Lemma~\ref{lem-imp1},~\ref{lem-imp2}, and~\ref{lem-km-sample3} all happen so that we do not need to repeatedly state the success probabilities. 


\begin{lemma}
\label{lem-km-1}
$Cost(S_{opt}, O^*)\leq (1+\delta)\frac{|S|}{n}(n-z)\big(\Delta^{-z}_2(P, O^*)+\xi  \mathcal{L}^2\big)$.
\end{lemma}
\begin{proof}
First, we have 
\begin{eqnarray}
Cost(S_{opt}, O^*)&=&\sum^k_{j=1}\sum_{q\in S\cap C^*_j}||q-o^*_j||^2\nonumber\\
&\leq&(1+\delta)\frac{|S|}{n}\sum^k_{j=1}\big(\sum_{p\in C^*_j}||p-o^*_j||^2+\xi |C^*_j| \mathcal{L}^2\big)
\end{eqnarray} 
by Lemma~\ref{lem-km-sample3}. Further, since $\sum^k_{j=1}\sum_{p\in C^*_j}||p-o^*_j||^2=(n-z)\Delta^{-z}_2(P, O^*)$ and $\sum^k_{j=1}|C^*_j|=n-z$, we know that $Cost(S_{opt}, O^*)$ is no larger than $(1+\delta)\frac{|S|}{n}(n-z)\big(\Delta^{-z}_2(P, O^*)+\xi  \mathcal{L}^2\big)$.
\qed
\end{proof}

\begin{lemma}
\label{lem-km-2}
$Cost(\tilde{S}_{opt}, H)\leq (2+2c)Cost(S_{opt}, O^*)$.
\end{lemma}
\begin{proof}
We fix a point $q\in S_{opt}$, and assume that the nearest neighbors of $q$ and $\tilde{q}$ in $H$ are $h_{j_q}$ and $h_{\tilde{j}_q}$, respectively. Then, we have
\begin{eqnarray}
||\tilde{q}-h_{\tilde{j}_q}||^2\leq ||\tilde{q}-h_{j_q}||^2\leq 2||\tilde{q}-q||^2+2||q-h_{j_q}||^2\label{for-km-2-1}
\end{eqnarray}
via the triangle inequality. Therefore,
\begin{eqnarray}
\sum_{q\in S_{opt}}||\tilde{q}-h_{\tilde{j}_q}||^2&\leq& 2\sum_{q\in S_{opt}}||\tilde{q}-q||^2+2\sum_{q\in S_{opt}}||q-h_{j_q}||^2,\nonumber\\
\Longrightarrow Cost(\tilde{S}_{opt}, H)&\leq& 2Cost(S_{opt}, O^*)+2Cost(S_{opt}, H).\label{for-km-2-2}
\end{eqnarray} 
Moreover, since $S_{opt}\subseteq S$ (because $S_{opt}=S\cap P_{opt}$) and $H$ yields a $c$-approximate clustering cost of the $(k+k')$-means clustering on $S$, we have 
\begin{eqnarray}
Cost(S_{opt}, H)\leq Cost(S, H)\leq c\cdot W,\label{for-km-2-3}
\end{eqnarray}
where $W$ is the optimal clustering cost of $(k+k')$-means clustering on $S$. Let $S'$ be the $k'$ farthest points of $S$ to $O^*$, then the set $O^*\cup S'$ also  
forms a solution for $(k+k')$-means clustering on $S$; namely, $S$ is partitioned into $k+k'$ clusters where each point of $S'$ is a cluster having a single point. Obviously, such a clustering yields a clustering cost $(|S|-k')\Delta^{-k'}_2(S, O^*)$. Consequently,
\begin{eqnarray}
W\leq (|S|-k')\Delta^{-k'}_2(S, O^*).\label{for-km-2-4}
\end{eqnarray}
Also, Lemma~\ref{lem-imp2}  shows that $S$ contains at most $k'$ points from $P\setminus P_{opt}$, {\em i.e.}, $|S_{opt}|\geq |S|-k'$. Thus, $Cost(S_{opt}, O^*)\geq (|S|-k')\Delta^{-k'}_2(S, O^*)$. Together with (\ref{for-km-2-2}), (\ref{for-km-2-3}), and (\ref{for-km-2-4}), we have $Cost(\tilde{S}_{opt}, H)\leq (2+2c)Cost(S_{opt}, O^*)$.
\qed
\end{proof}

\begin{lemma}
\label{lem-km-3}
$Cost(\tilde{P}_{opt}, H)\leq \frac{1}{1-\delta}\frac{n}{|S|}Cost(\tilde{S}_{opt}, H)$.
\end{lemma}
\begin{proof}
From the constructions of $\tilde{P}_{opt}$ and $\tilde{S}_{opt}$, we know that they are overlapping points locating at $\{o^*_1, \cdots, o^*_k\}$. From Lemma~\ref{lem-imp1}, we know $|S\cap C^*_j|\geq (1-\delta)\frac{|C^*_j|}{n}|S|$, {\em i.e.}, $|C^*_j|\leq \frac{1}{1-\delta}\frac{n}{|S|}|S\cap C^*_j|$ for $1\leq j\leq k$. Overall, we have 
$Cost(\tilde{P}_{opt}, H)=\sum^k_{j=1}|C^*_j|\big(dist(o^*_j, H)\big)^2$ that is at most 
$\frac{1}{1-\delta}\frac{n}{|S|}\sum^k_{j=1}|S\cap C^*_j|\big(dist(o^*_j, H)\big)^2= \frac{1}{1-\delta}\frac{n}{|S|}Cost(\tilde{S}_{opt}, H)$.
%
%
%
%
\qed
\end{proof}

Now, we are ready to prove Theorem~\ref{the-km}. Note that $(n-z)\Delta^{-z}_2(P, H)$ actually is the $|H|$-means clustering cost of $P$ by removing the farthest $z$ points to $H$, and $|P_{opt}|=n-z$. So we have $(n-z)\Delta^{-z}_2(P, H)\leq Cost(P_{opt}, H)$. Further, by using a similar manner of (\ref{for-km-2-2}), we have $Cost(P_{opt}, H)\leq 2Cost(P_{opt}, O^*)+2Cost(\tilde{P}_{opt}, H)$.
 Therefore,
\begin{eqnarray}
\Delta^{-z}_2(P, H)&\leq&\frac{1}{n-z}Cost(P_{opt}, H)\leq \frac{2}{n-z}\Big(Cost(P_{opt}, O^*)+Cost(\tilde{P}_{opt}, H)\Big)\nonumber\\
&\leq& \frac{2}{n-z}\Big(Cost(P_{opt}, O^*)+\frac{1}{1-\delta}\frac{n}{|S|}Cost(\tilde{S}_{opt}, H) \Big) \label{for-the-km-f1}\\
&\leq& \frac{2}{n-z}\Big(Cost(P_{opt}, O^*)+\frac{1}{1-\delta}\frac{n}{|S|}(2+2c)Cost(S_{opt}, O^*) \Big) \label{for-the-km-f2}\\
&\leq& \frac{2}{n-z}\Big(Cost(P_{opt}, O^*)+\frac{1+\delta}{1-\delta}(2+2c)(n-z)\big(\Delta^{-z}_2(P, O^*)+\xi  \mathcal{L}^2\big)\Big) \label{for-the-km-f3} \\
&=&\big(2+(4+4c)\frac{1+\delta}{1-\delta}\big)\Delta^{-z}_{2}(P, O^*)+(4+4c)\frac{1+\delta}{1-\delta}\xi \mathcal{L}^2,  \label{for-the-km-f4}
\end{eqnarray}
where (\ref{for-the-km-f1}), (\ref{for-the-km-f2}), and (\ref{for-the-km-f3}) come from Lemma~\ref{lem-km-3}, \ref{lem-km-2}, and \ref{lem-km-1} respectively, and (\ref{for-the-km-f4}) comes from the fact $ Cost(P_{opt}, O^*)=(n-z)\Delta^{-z}_2(P, O^*)$. The success probability $(1-\eta)^3$ comes from Lemma~\ref{lem-km-sample3} and Lemma~\ref{lem-imp2}  (note that Lemma~\ref{lem-km-sample3} already takes into account of the success probability of Lemma~\ref{lem-imp1} ). Thus, we obtain Theorem~\ref{the-km}.

\subsection{Proof of Theorem~\ref{the-km2}}
\label{sec-proof-km2}

Suppose the $k$ clusters of $S$ obtained in Step 2 of Algorithm~\ref{alg-km2} are $S_1, S_2, \cdots, S_k$, and thus the inliers $S_{in}=\cup^k_{j=1}S_j$. In the following lemmas, we always assume that the events mentioned in Lemma~\ref{lem-imp1},~\ref{lem-imp2}, and~\ref{lem-km-sample3} all happen so that we do not need to repeatedly state the success probabilities.

\begin{lemma}
\label{lem-km2-1}
$\frac{|C^*_j|}{|C^*_j\cap S_{in}|}\leq \frac{n}{|S|}\frac{t}{(t-1)(1-\delta)}$ for each $1\leq j\leq k$. 
\end{lemma}
\begin{proof}
Since $z'=\frac{1}{\eta}\frac{\epsilon_2}{k}|S|$ and $|S\cap C^*_j|\geq (1-\delta)\frac{|C^*_j|}{n}|S|$ for each $1\leq j\leq k$ (by Lemma~\ref{lem-imp1}), we have 
\begin{eqnarray}
|S_{in}\cap C^*_j|&\geq& |S\cap C^*_j|-z'\geq (1-\delta)\frac{|C^*_j|}{n}|S|-\frac{1}{\eta}\frac{\epsilon_2}{k}|S|\nonumber\\
&=&\Big(1-\frac{1}{\eta(1-\delta)}\frac{\epsilon_2}{k}\frac{n}{|C^*_j|}\Big)(1-\delta)\frac{|C^*_j|}{n}|S|\nonumber\\
&\geq&\Big(1-\frac{1}{\eta(1-\delta)}\frac{\epsilon_2}{\epsilon_1}\Big)(1-\delta)\frac{|C^*_j|}{n}|S|=(1-\frac{1}{t})(1-\delta)\frac{|C^*_j|}{n}|S|,
\end{eqnarray}
where the last inequality comes from $|C^*_j|\geq \frac{\epsilon_1}{k}n$. Thus 
$\frac{|C^*_j|}{|C^*_j\cap S_{in}|}\leq \frac{n}{|S|}\frac{t}{(t-1)(1-\delta)}$.  
\qed
\end{proof}

\begin{lemma}
\label{lem-km2-2}
$Cost(S_{in}\cap P_{opt}, H)\leq (1+\delta)\frac{|S|}{n}\cdot c\cdot \big(Cost(P_{opt}, O^*)+(n-z)\cdot \xi \mathcal{L}^2\big)$. 
\end{lemma}
\begin{proof}
Since $z'\geq |S\setminus P_{opt}|=|S|-|S_{opt}|$, we have
\begin{eqnarray}
(|S|-z')\Delta^{-z'}_2(S, O^*)\leq Cost(S_{opt}, O^*). \label{for-km2-2-2}
\end{eqnarray}
Because $H$ is a $c$-approximation on $S$, 
\begin{eqnarray}
\Delta^{-z'}_2(S, H)\leq c\cdot \Delta^{-z'}_2(S, O^*)\leq \frac{c}{|S|-z'}Cost(S_{opt}, O^*), \label{for-km2-2-3}
\end{eqnarray}
where the last inequality comes from (\ref{for-km2-2-2}). Therefore,
\begin{eqnarray}
Cost(S_{in}\cap P_{opt}, H)&\leq& Cost(S_{in}, H)=(|S|-z')\Delta^{-z'}_2(S, H)\nonumber\\
&\leq&c\cdot Cost(S_{opt}, O^*)\nonumber\\
&\leq& (1+\delta)\frac{|S|}{n}\cdot c\cdot (n-z)\big(\Delta^{-z}_2(P, O^*)+\xi  \mathcal{L}^2\big),
\end{eqnarray}
where the second and third inequalities comes from (\ref{for-km2-2-3}) and Lemma~\ref{lem-km-1}, respectively.
\qed
\end{proof}
Since $S_{in}\cap P_{opt}\subseteq S_{opt}$, we immediately have the following lemma via Lemma~\ref{lem-km-1}.
\begin{lemma}
\label{lem-km2-3}
$Cost(S_{in}\cap P_{opt}, O^*)\leq (1+\delta)\frac{|S|}{n}(n-z)\big(\Delta^{-z}_2(P, O^*)+\xi  \mathcal{L}^2\big)$.
\end{lemma}

For convenience, let $S'_{in}=S_{in}\cap P_{opt}$. Using the same manner of (\ref{for-km-2-2}), we have 
\begin{eqnarray}
Cost(\tilde{S}'_{in}, H)&\leq& 2 Cost(S'_{in}, O^*)+2 Cost(S'_{in}, H);\label{for-km2-2-4}\\
Cost(P_{opt}, H)&\leq& 2Cost(P_{opt}, O^*)+2Cost(\tilde{P}_{opt}, H). \label{for-km2-2-5}
\end{eqnarray}
Also, because $Cost(\tilde{P}_{opt}, H)=\sum^k_{j=1}|C^*_j|(dist(o^*_j, H))^2$ and $Cost(\tilde{S}'_{in}, H)=\sum^k_{j=1}|C^*_j\cap S_{in}|(dist(o^*_j, H))^2$, we have 
\begin{eqnarray}
Cost(\tilde{P}_{opt}, H)&\leq& \max_{1\leq j\leq k}\frac{|C^*_j|}{|C^*_j\cap S_{in}|}\cdot Cost(\tilde{S}'_{in}, H)\nonumber\\
&\leq&\frac{n}{|S|}\frac{t}{(t-1)(1-\delta)}\cdot Cost(\tilde{S}'_{in}, H) \label{for-km2-2-6}
\end{eqnarray}
where the last inequality comes from Lemma~\ref{lem-km2-1}. From (\ref{for-km2-2-4}), (\ref{for-km2-2-5}), and (\ref{for-km2-2-6}), we have
\begin{eqnarray}
Cost(P_{opt}, H)\leq 2Cost(P_{opt}, O^*) +\frac{4n}{|S|}\frac{t}{(t-1)(1-\delta)}\big(Cost(S'_{in}, O^*)+Cost(S'_{in}, H)\big)\label{for-km2-2-7}
\end{eqnarray}
Combining Lemma~\ref{lem-km2-2} and~\ref{lem-km2-3}, we can obtain Theorem~\ref{the-km2} from (\ref{for-km2-2-7}) by simple calculation.

\section{Success Probability and Clustering Memberships}
\label{sec-boost}
 The parameter $\eta$ determines the success probabilities of our algorithms.
In particular, as mentioned in Remark~\ref{rem-2}, we cannot set $\eta$ too small to guarantee ``$\frac{\epsilon_1}{\epsilon_2}>\frac{1}{\eta(1-\delta)}$'' in Algorithm~\ref{alg-kc2} (and similarly in Algorithm~\ref{alg-km2}). To satisfy this requirement, we need to set $\eta$ large enough and therefore the success probability could be low. In fact, we can run the algorithm multiple times so as to achieve a higher success probability; for example, if $\eta=0.8$ and we run the algorithm $50$ times, the success probability will be $1-(1-(1-0.8)^2)^{50}\approx 87\%$. Suppose we run the algorithm (Algorithm~~\ref{alg-kc2} or~\ref{alg-km2}) $m>1$ times and let $H_1, \cdots, H_{m}$ be the set of output candidates. 
The remaining issue is that how to select the one achieving the smallest objective value among all the candidates.

A simple way is to directly scan the whole dataset in one-pass. When reading a point $p$ from $P$, we calculate its distance to all the candidates, {\em i.e.,} $dist(p, H_1), \cdots, dist(p, H_m)$; after scanning the whole dataset, we have calculated the clustering costs $\Delta^{-z}_\infty(P, H_l)$ ({\em resp.,} $\Delta^{-z}_1(P, H_l)$ and $\Delta^{-z}_2(P, H_l)$) for $1\leq l\leq m$ and return the best one. Moreover, a by-product of this procedure is that we can determine the clustering memberships of data points simultaneously. When calculating $dist(p, H_l)$ for $1\leq l\leq m$, we record the index of its nearest cluster center in $H_l$; finally, we return the corresponding clustering memberships after selecting the best candidate. 

We are aware of the sampling method proposed by Meyerson {\em et al.}~\cite{meyerson2004k} for estimating $k$-median clustering cost; but it will induce an error on the number of outliers for our clustering with outliers problems. As mentioned in Section~\ref{sec-ourresult}, the sampling based ideas in~\cite{charikar2003better,huang2018epsilon,DBLP:conf/esa/DingYW19} also have the same issue. 
 \section{An Example for Theorem~\ref{the-kc1} and \ref{the-km}}
\label{sec-example}
 
We construct the following instance for $k$-center clustering with outliers. Let $P$ be an $(\epsilon_1, \epsilon_2)$-significant instance in $\mathbb{R}^D$, where each optimal cluster $C^*_j$ is a set of $|C^*_j|$ overlapping points located at its cluster center $o^*_j$ for $1\leq j\leq k$. Let $x>0$, and we assume
\begin{eqnarray}
||o^*_{j_1}-o^*_{j_2}||&\geq&x, \hspace{0.3in} \forall j_1\neq j_2;\nonumber\\
||q_1-q_2||&\geq& x,\hspace{0.3in}  \forall q_1, q_2\in P\setminus P_{opt};\nonumber\\
||o^*_j-q||&\geq& x, \hspace{0.3in} \forall 1\leq j\leq k, q\in P\setminus P_{opt}.
\end{eqnarray}

Obviously, the optimal radius $r_{opt}$ of $P$ is equal to $0$. Suppose we obtain a sample $S$ satisfying $S\cap C^*_j\neq\emptyset$ for any $1\leq j\leq k$ and $|S\setminus P_{opt}|=\frac{1}{\eta}\frac{\epsilon_2}{k}|S|=k'$. Given a number $k''<k'$, we run $(k+k'')$-center clustering on $S$. Since the points of $S$ take $k+k'$ distinct locations in the space,  any $(k+k'')$-center clustering on $S$ will result in a radius  (at least $x/2$) larger than $0$; thus the approximation ratio is equal to $+\infty$. 
So, the value of $k'$ cannot be further reduced in Theorem~\ref{the-kc1}. It is easy to see that this instance $P$ also can be used to show that $k'$ should be at least $\frac{1}{\eta}\frac{\epsilon_2}{k}|S|$ in Theorem~\ref{the-km}.

%
 \vspace{-0.15in}
\section{Experiments}
\label{sec-exp}
\vspace{-0.1in}

%

All the experimental results were obtained on a Windows workstation with 2.8GHz Intel(R) Core(TM) i5-840 and 8GB main memory; the algorithms are implemented in Matlab R2018a.
To evaluate the performance, we use several baseline algorithms including the non-uniform sampling approaches mentioned in Section~\ref{sec-ourresult} (we only consider the setting with single machine in this paper). For $k$-center clustering with outliers, we consider four existing algorithms:  the $3$-approximation \textsc{\textbf{Charikar}}~\cite{charikar2001algorithms},  the $(4+\epsilon)$-approximation \textbf{\textsc{MK}}~\cite{mccutchen2008streaming}, the $13$-approximation \textbf{\textsc{Malkomes}}~\cite{malkomes2015fast}, and the greedy algorithm
\textbf{\textsc{DYW}}~\cite{DBLP:conf/esa/DingYW19} (as the non-uniform sampling approach). The distributed algorithm \textsc{Malkomes} partitions the dataset into $m\geq 1$ parts and processes each part separately; to make a fair comparison, we set $m=1$ for \textsc{Malkomes} in our experiments. In our Algorithm~\ref{alg-kc2}, we apply \textsc{MK} as the subroutine in Step 2. 


For $k$-means/median clustering with outliers, we consider the heuristic algorithm \textbf{\textsc{$k$-means$--$}}~\cite{chawla2013k} and two non-uniform sampling methods: the local search algorithm \textbf{\textsc{LocalSearch}} with $k$-means++~\cite{gupta2017local}, and the recent data summary based algorithm \textbf{\textsc{DataSummary}}~\cite{DBLP:conf/nips/ChenA018}. In our Algorithm~\ref{alg-km} ({\em resp.,} Algorithm~\ref{alg-km2}), we apply the $k$-means++~\cite{arthur2007k} ({\em resp.,} \textsc{$k$-means$--$}) as the subroutine in Step 2.

\textbf{Datasets.} We generate the synthetic datasets in $\mathbb{R}^{100}$, and set $n=10^5$, $z=2\%n$, and $k=8$. We randomly generate k points as the cluster centers inside a hypercube of side length $400$; around each center, we generate a cluster of points following a Gaussian distribution with standard deviation $\sqrt{1000}$; we keep the total number of points  to be $n-z$; to study the performance of our algorithms with respect to the ratio $\frac{\epsilon_1}{\epsilon_2}$, we vary the size of the smallest cluster appropriately for each synthetic dataset; finally, we uniformly generate $z$ outliers at random outside the minimum enclosing balls of these $k$ clusters. 

We choose $4$ real datasets from {\em UCI machine learning repository}~\cite{Dua:2019}. \textbf{Covertype} has $6$ clusters with $575000$ points in $\mathbb{R}^{54}$; \textbf{Kddcup} has $3$ clusters with $494021$ points in $\mathbb{R}^{38}$; \textbf{Poking Hand} has $4$ clusters with $995000$ points in $\mathbb{R}^{10}$; \textbf{Shuttle} has $3$ clusters with $59000$ points in $\mathbb{R}^9$.
Each  dataset also has some tiny clusters with total size $<1\%$, and we view them as outliers. To keep the fraction of outliers to be $1\%$, we add extra outliers outside the enclosing balls of the clusters as we did for  the synthetic datasets.

\begin{figure}
	\begin{center}
			 \vspace{-0.1in}
		\includegraphics[width=0.24\textwidth]{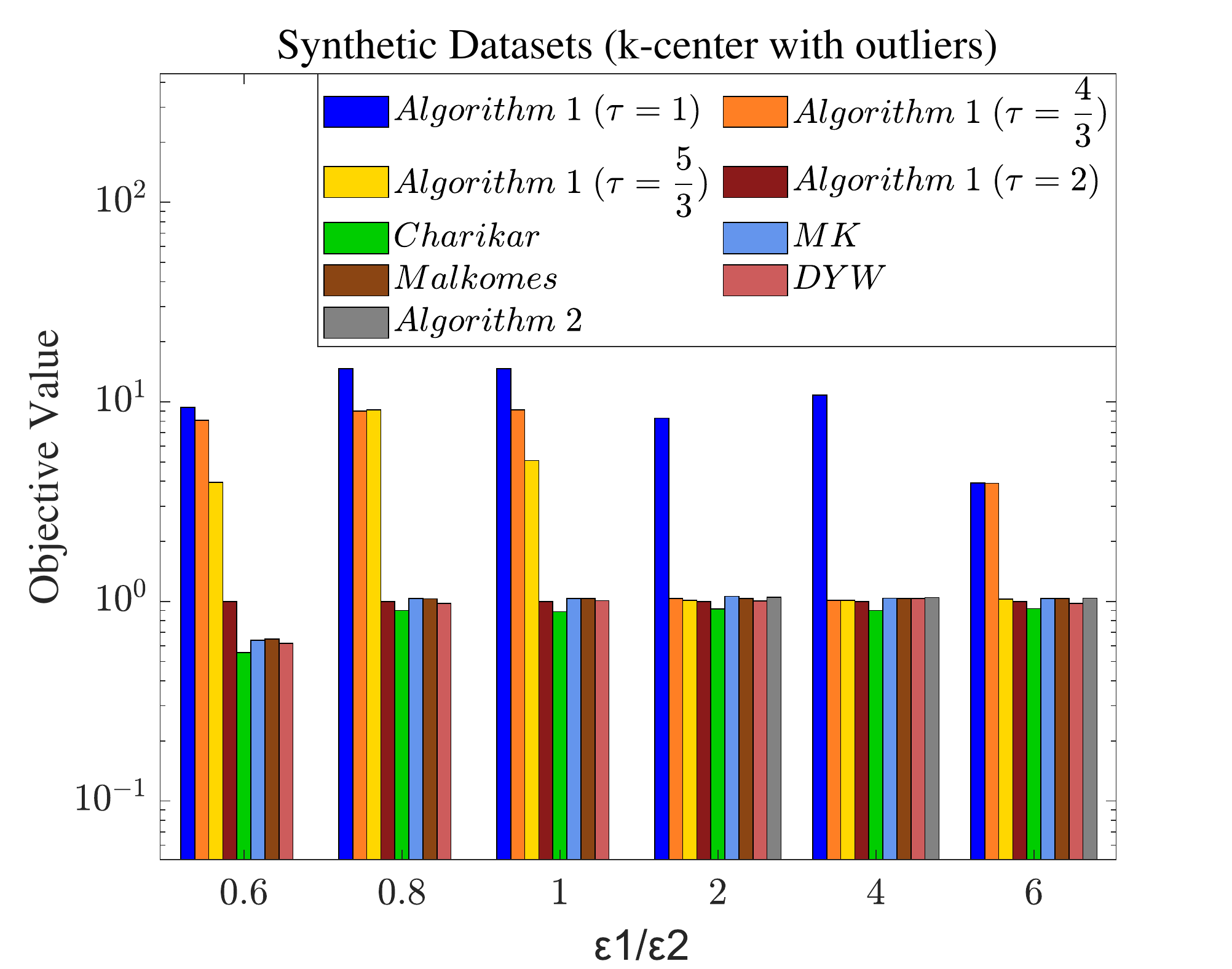}  
		\includegraphics[width=0.24\textwidth]{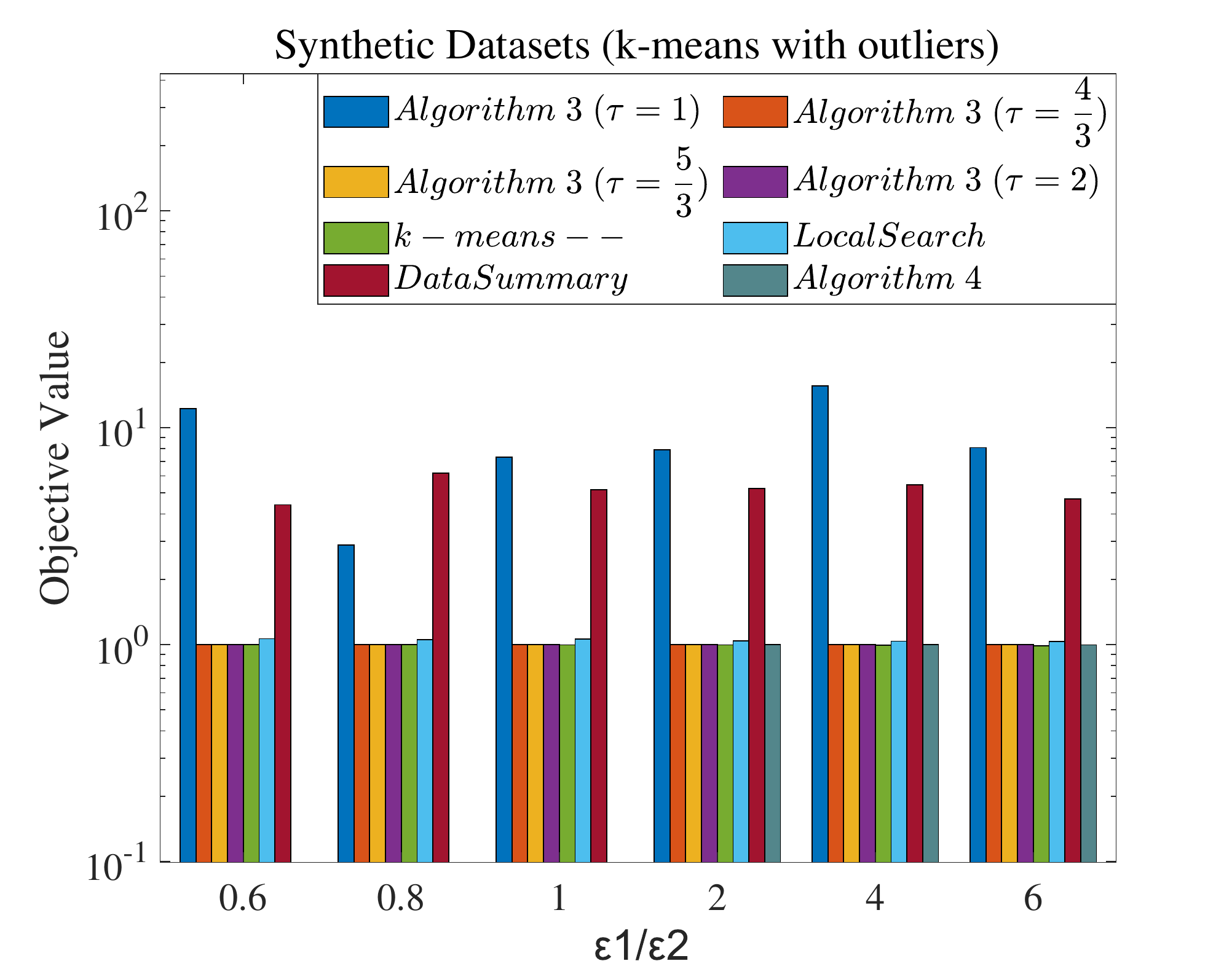} 
		\includegraphics[width=0.24\textwidth]{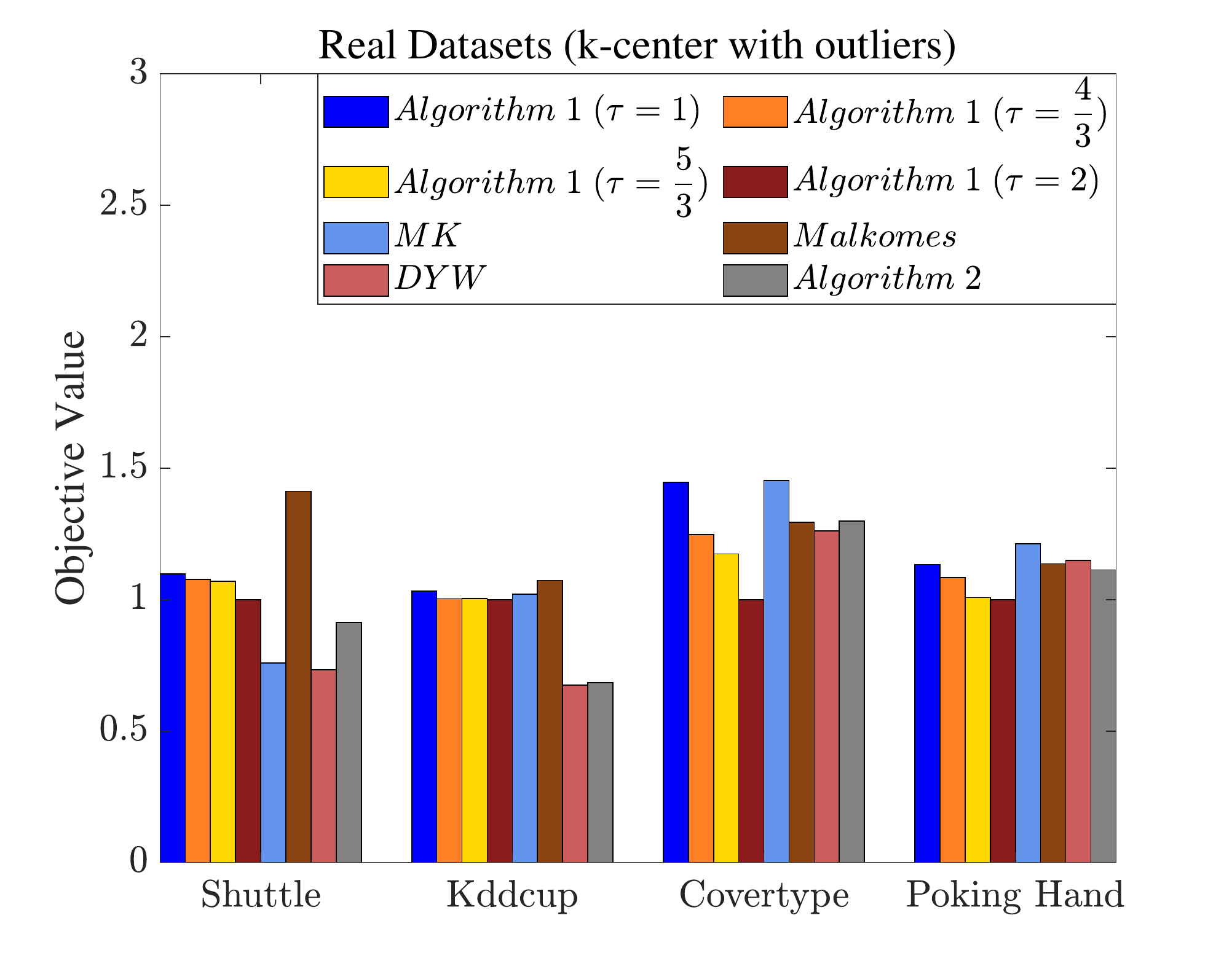}  
		\includegraphics[width=0.24\textwidth]{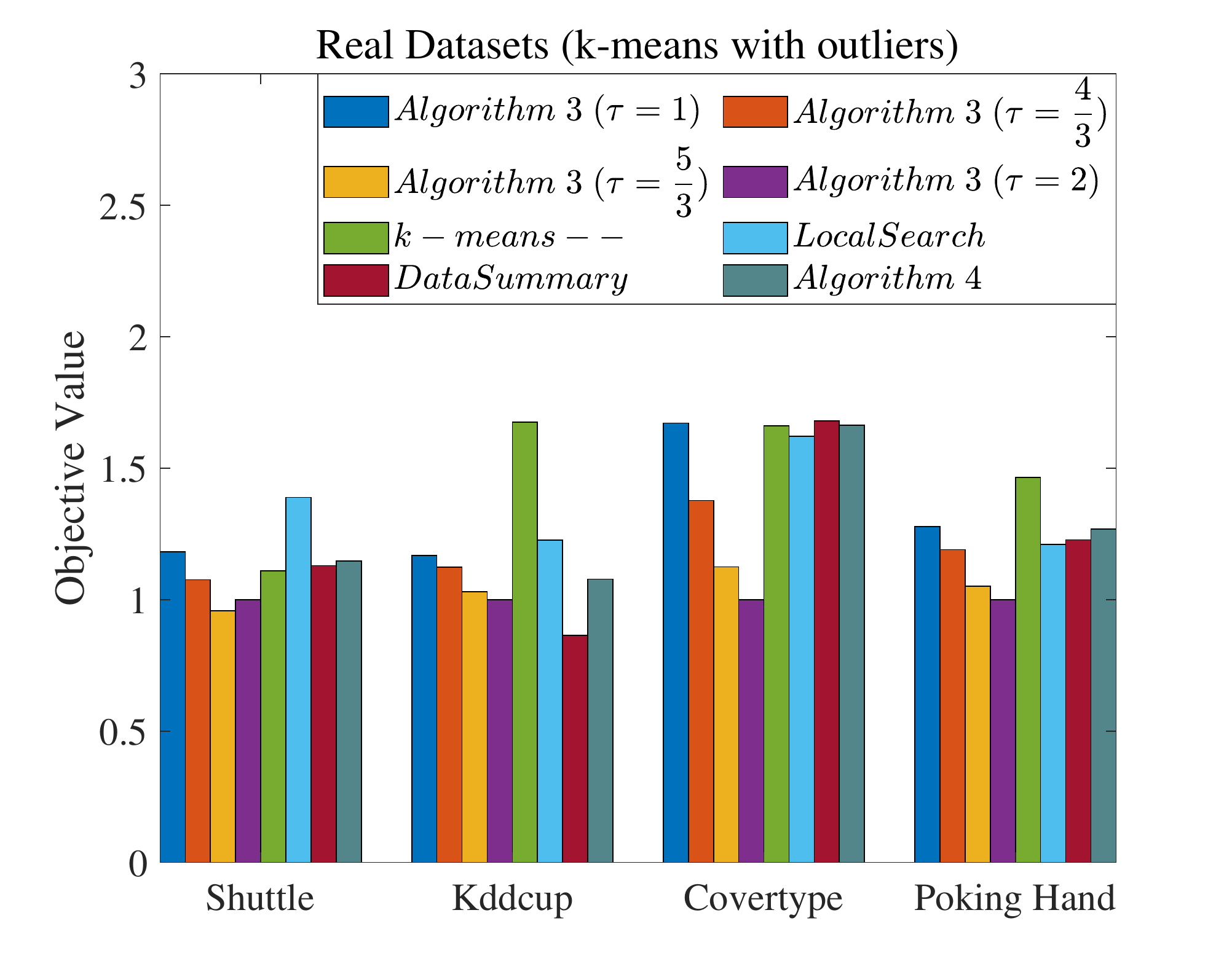}   
		 \vspace{-0.1in}
		\caption{The normalized objective values on the synthetic and real  datasets.}     
		\label{fig-exp-objective}
	\end{center}
	 \vspace{-0.3in}
\end{figure}

\begin{figure}
	\begin{center}
			 \vspace{-0.3in}
		\includegraphics[width=0.24\textwidth]{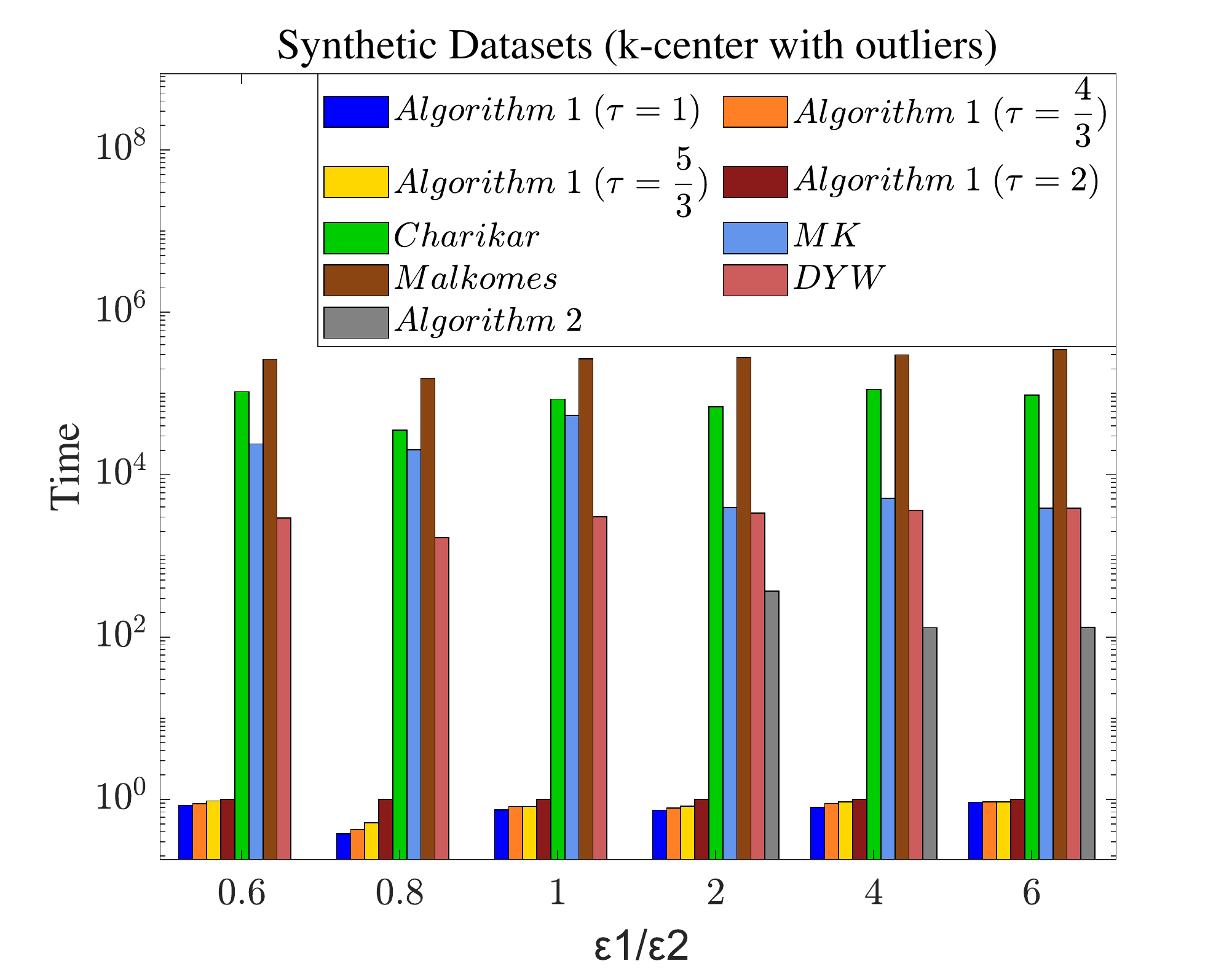}  
		\includegraphics[width=0.24\textwidth]{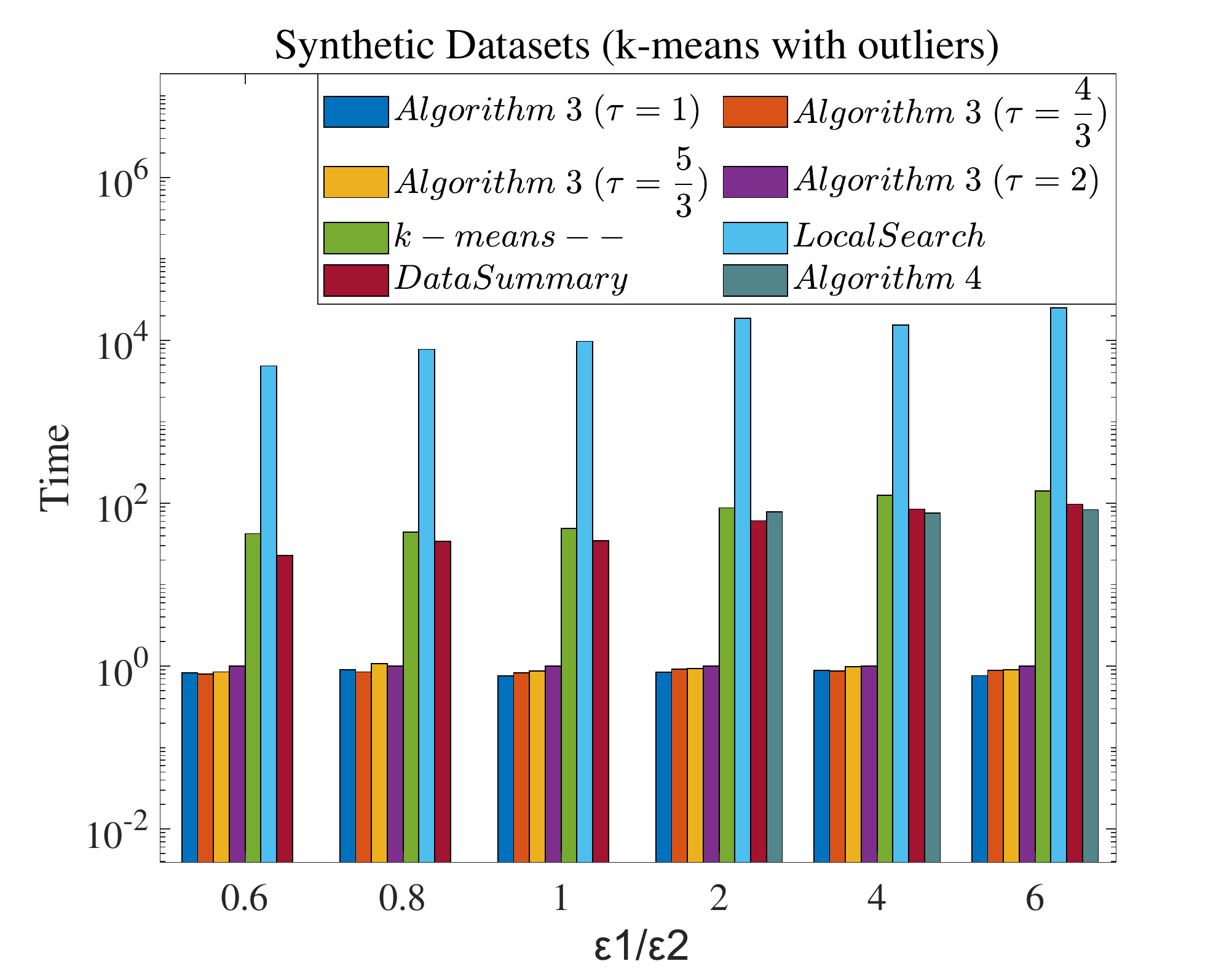} 
		\includegraphics[width=0.24\textwidth]{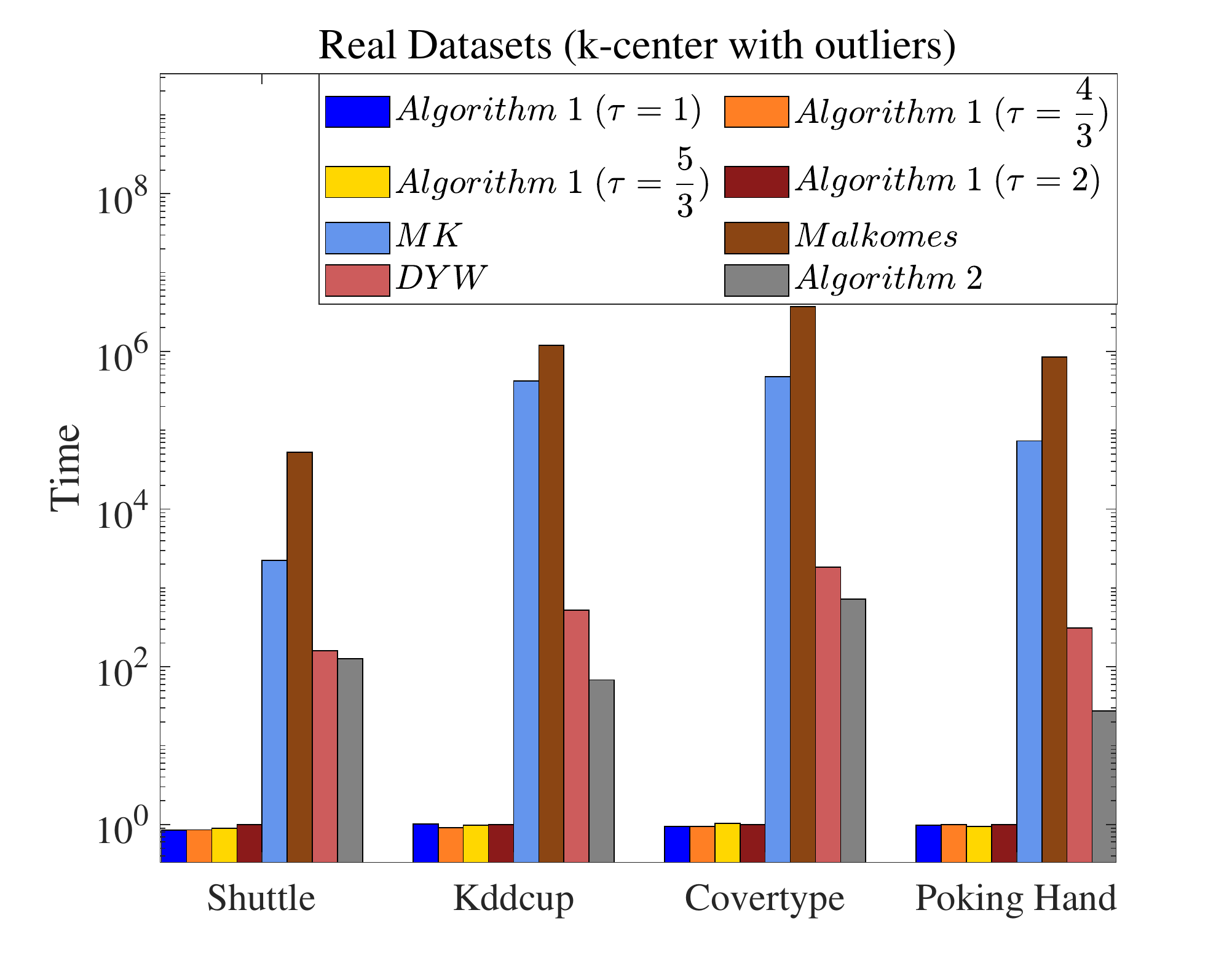}  
		\includegraphics[width=0.24\textwidth]{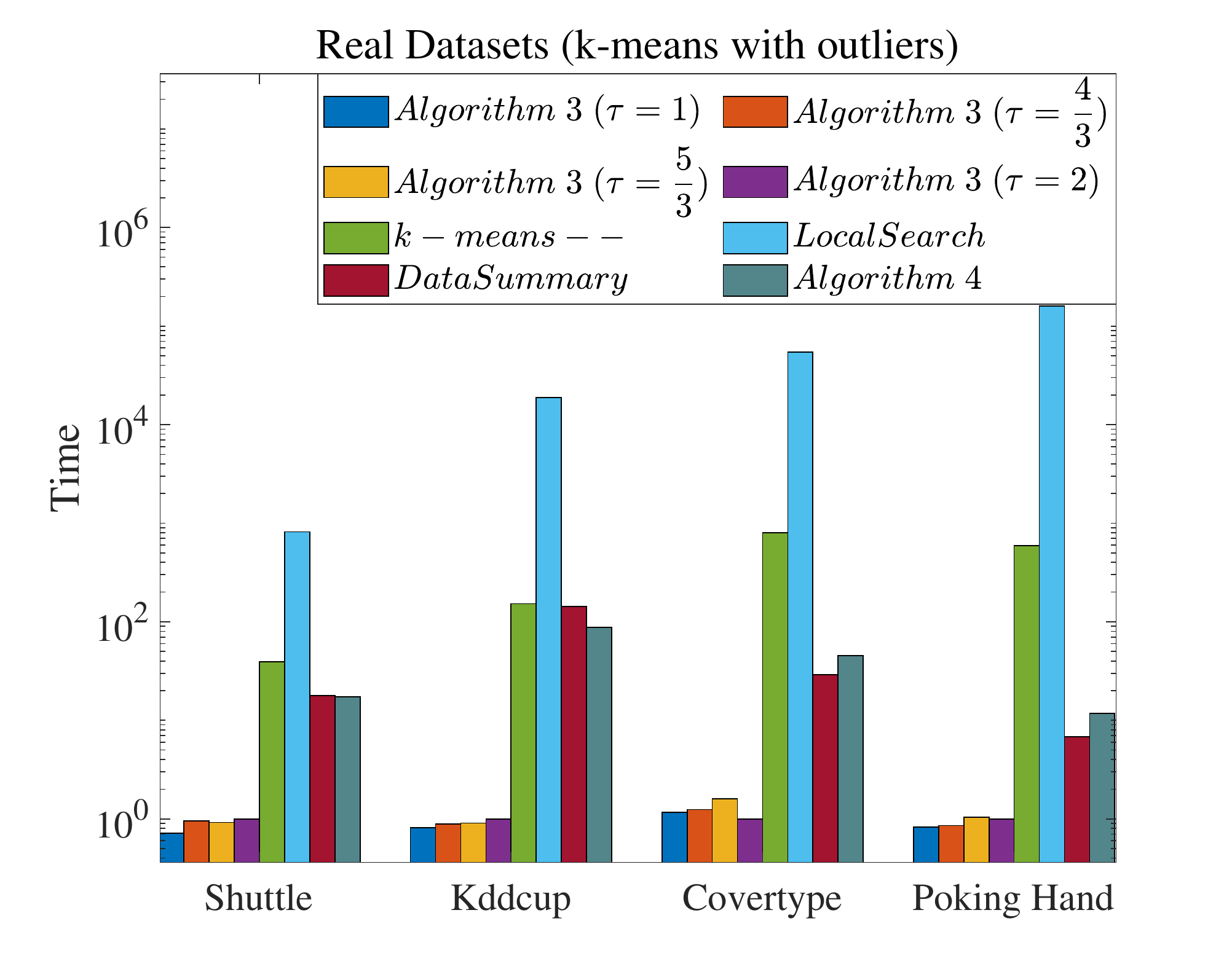}  
				 \vspace{-0.1in}
		  \caption{The normalized running times on the synthetic and real datasets.}     
		\label{fig-exp-time}
	\end{center}
			 \vspace{-0.3in}
\end{figure}

\textbf{Settings.} For each dataset, we run our algorithms and the baseline algorithms $20$ trials and report the average results. For our algorithms, it is not quite convenient to set the values for the parameters $(\eta,\xi,\delta)$ in practice. Instead, it is more intuitive to directly set the sample size $|S|$ and $k'$ for Algorithm~\ref{alg-kc1} and~\ref{alg-km} ({\em resp.,}  $z'$ for Algorithm~\ref{alg-kc2} and~\ref{alg-km2}); actually, we only need these two numbers $(|S|, k')$ ({\em resp.,} $(|S|, z')$) to implement our algorithms. In our experiments on both synthetic and real datasets, we set the sample size $|S|$ to be $2\% n$; for completeness, we also investigate the stability with varying $|S|/n$ in another experiment below. Algorithm~\ref{alg-kc1} and~\ref{alg-km} both output  $k+k'$ cluster centers, and we define the ratio $\tau=\frac{k+k'}{k}$.
The value $k'=\frac{1}{\eta}\frac{\epsilon_2}{k}|S|$ could be large. Though Section~\ref{sec-example} indicates that $k'$ cannot be reduced with respect to the worst case, 
we do not strictly follow this (overly conservative) theoretical value in our experiments. Instead, we keep the ratio $\tau$ to be $1, \frac{4}{3}, \frac{5}{3}$, and $2$ ({\em i.e.,} run the algorithms~\cite{G85,arthur2007k} $k+k'=\tau k$ steps). 
For Algorithm~\ref{alg-kc2} ({\em resp.,} Algorithm~\ref{alg-km2}), we set $z'=2\frac{\epsilon_2}{k}|S|$ that is $2$ times the expected number of outliers in $S$; we run the algorithm $10$ times for each instance and select the best candidate by scanning the whole dataset in one-pass as discussed in Section~\ref{sec-boost} (we count the running time of the whole process). Further, we conduct another two experiments below to observe the influence of $z'$ and the stabilities of the results returned by Algorithm~\ref{alg-kc2} and~\ref{alg-km2} (if we just run each of them by one time).


%
%
%

%

\textbf{Objective value and running time.} The obtained objective values of our and the baseline algorithms are shown in Figure~\ref{fig-exp-objective}; the running times are shown in Figure~\ref{fig-exp-time}; due to the space limit, we leave the experimental results of $k$-median clustering with outliers to appendix (Section~\ref{sec-exp-kmed}). For $k$-center clustering with outliers, our algorithms (Algorithm~\ref{alg-kc2} and Algorithm~\ref{alg-kc1} with $\tau=2$) and the four baseline algorithms achieve similar objective values for most of the instances (we run Algorithm~\ref{alg-kc2} on the synthetic datasets with $\frac{\epsilon_1}{\epsilon_2}>1$ only; we do not run \textsc{Charikar} on the real datasets due to its high complexity). 
Moreover, the running times of our Algorithm~\ref{alg-kc1} and~\ref{alg-kc2} are significantly lower comparing with the baseline algorithms. 
For $k$-means clustering with outliers, Algorithm~\ref{alg-km} (except for the setting with $\tau=1$) and Algorithm~\ref{alg-km2} can achieve the results close to the best of the three baseline algorithms. \textsc{DataSummary} and Algorithm~\ref{alg-km2} achieve comparable running times, but Algorithm~\ref{alg-km2}  outperforms \textsc{DataSummary} with respect to the objective values on the synthetic datasets.


%
%
%

\textbf{The sample size $|S|$.} We also study the influence of the sample size $|S|$ on the experimental performances of our algorithms. 
 We vary the size $|S|$ from $2\%n$ to $10\%n$ and run our algorithms on the synthetic datasets ($\tau=2$ for Algorithm~\ref{alg-kc1} and~\ref{alg-km}). We show the results in Figure~\ref{fig-exp-size}. We can see that their performances stay stable when varying $|S|$. 

\begin{figure}
	\begin{center}
			 \vspace{-0.2in}
		\includegraphics[width=0.24\textwidth]{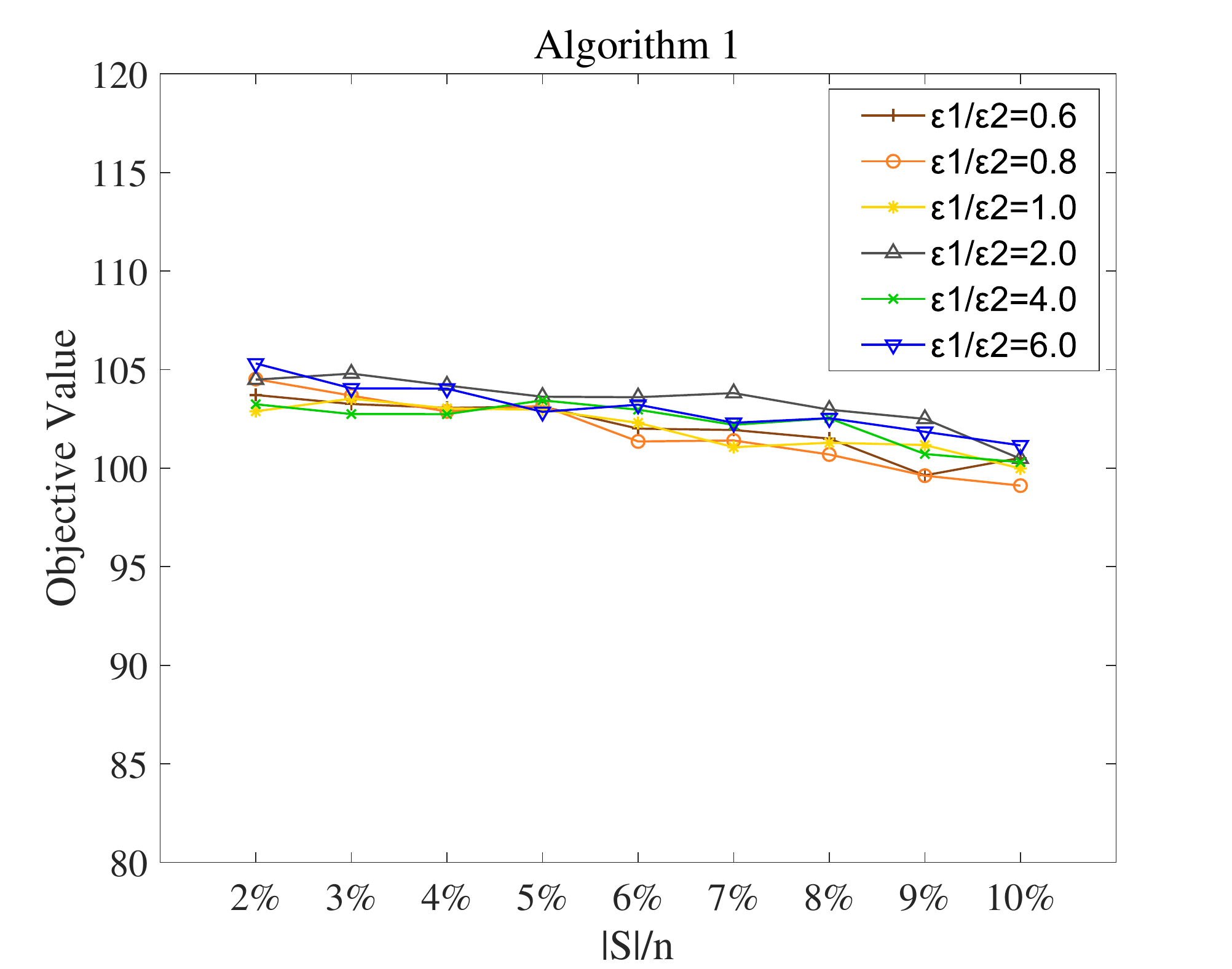}  
		\includegraphics[width=0.24\textwidth]{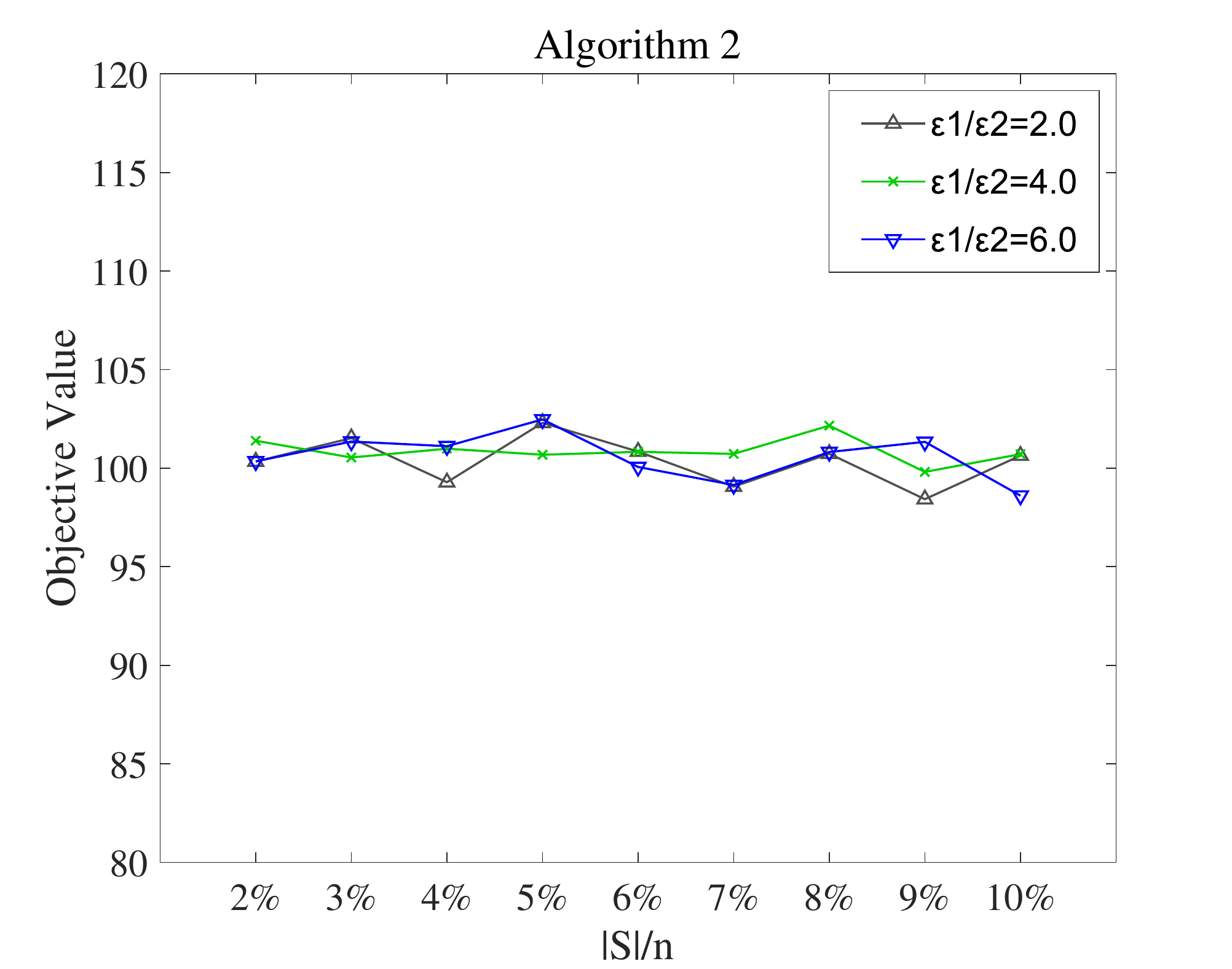} 
		\includegraphics[width=0.24\textwidth]{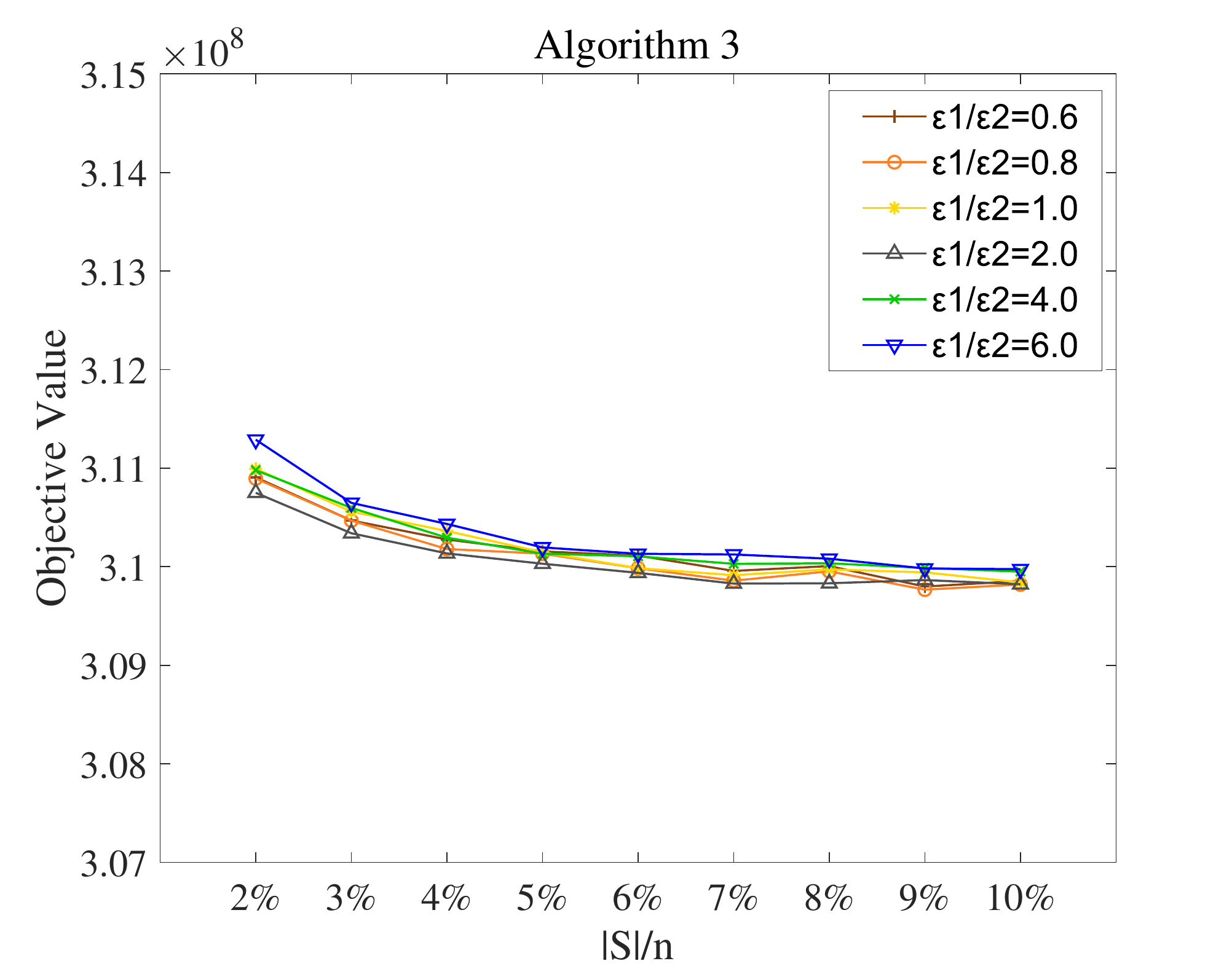}  
		\includegraphics[width=0.24\textwidth]{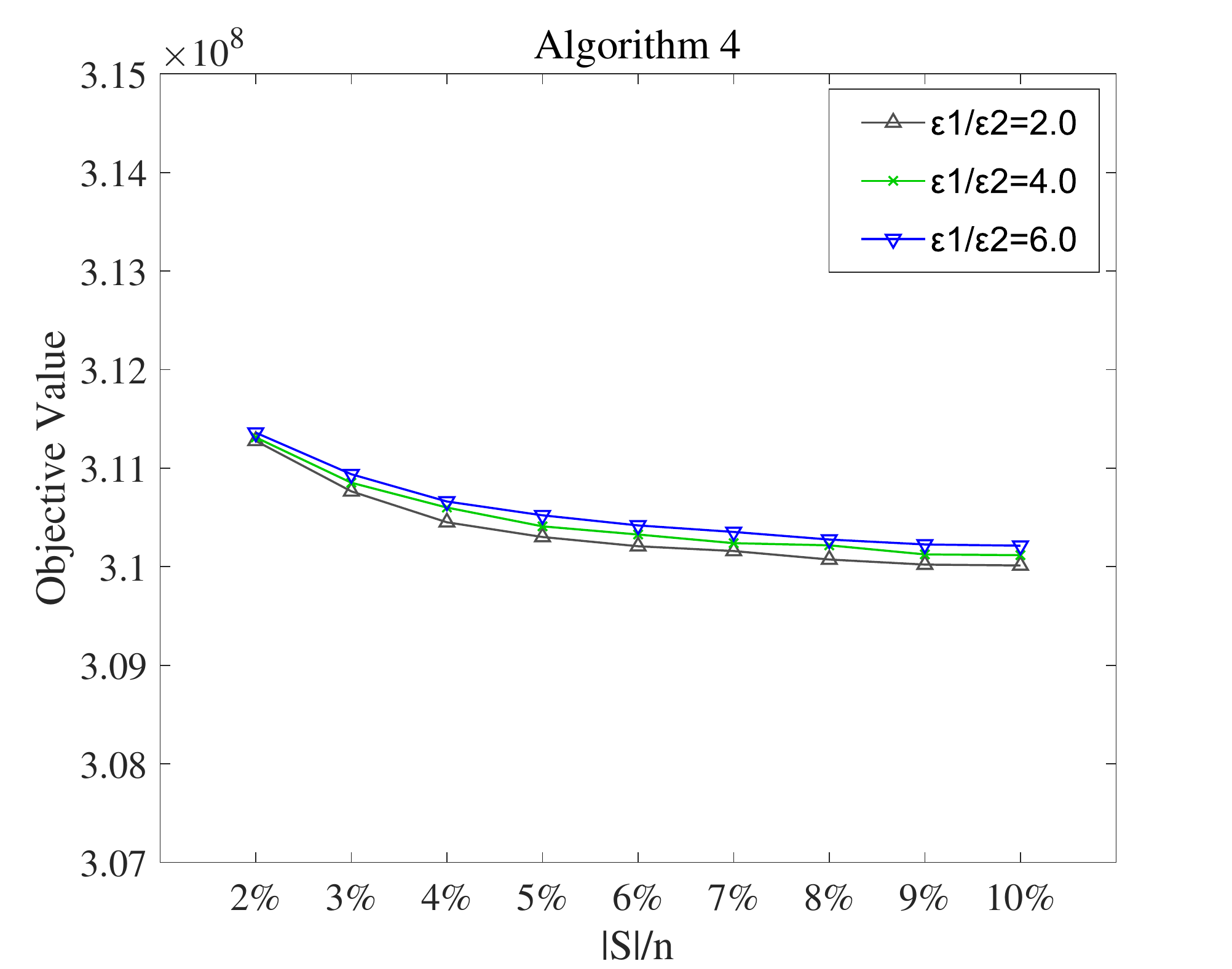}  
				 \vspace{-0.1in}
		  \caption{The performance with varying the sample size $|S|$.}     
		\label{fig-exp-size}
	\end{center}
			 \vspace{-0.3in}
\end{figure}

%
%

\begin{figure}
	\begin{center}
		 \vspace{-0.3in}
		  \includegraphics[width=0.24\textwidth]{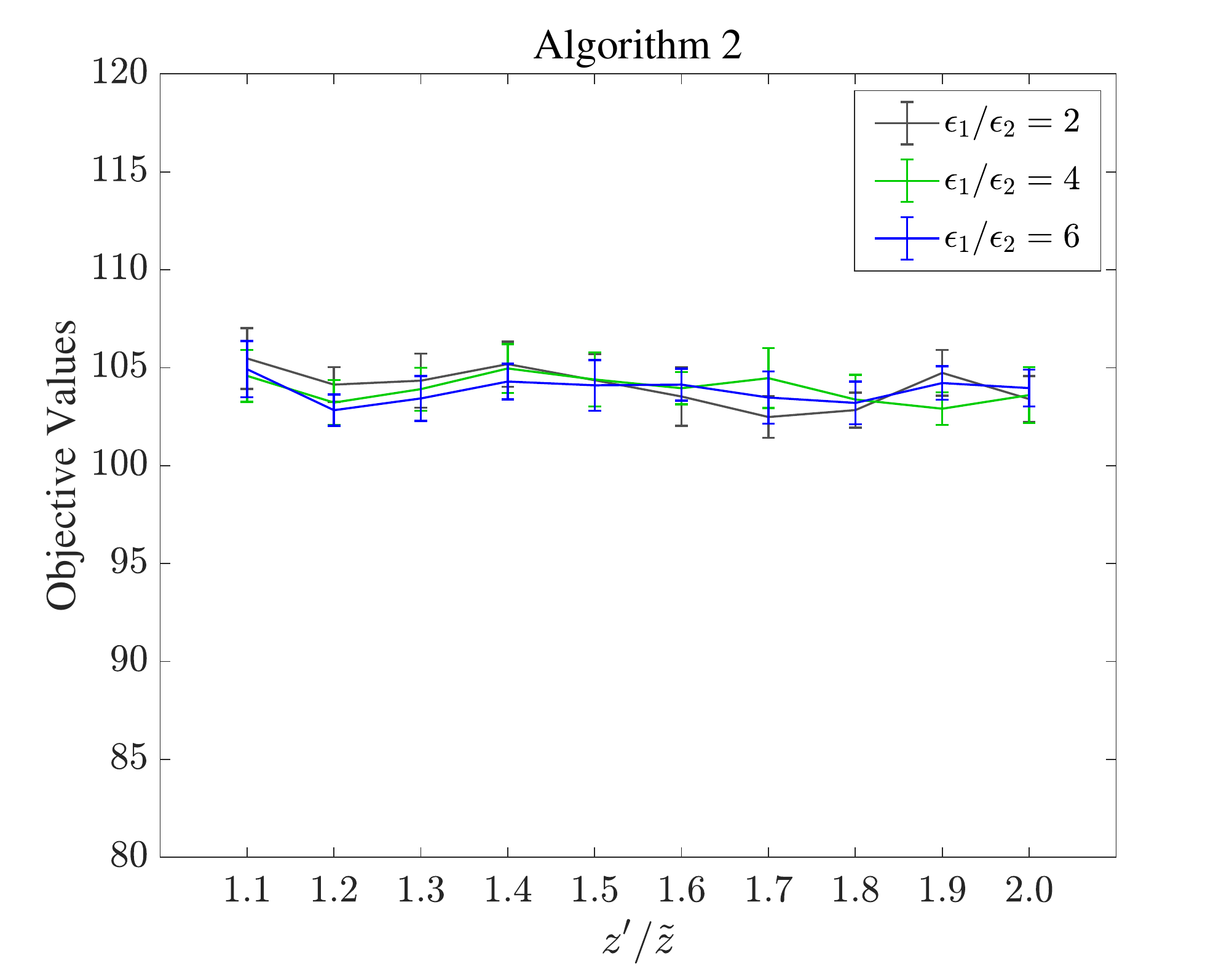}
                  \includegraphics[width=0.24\textwidth]{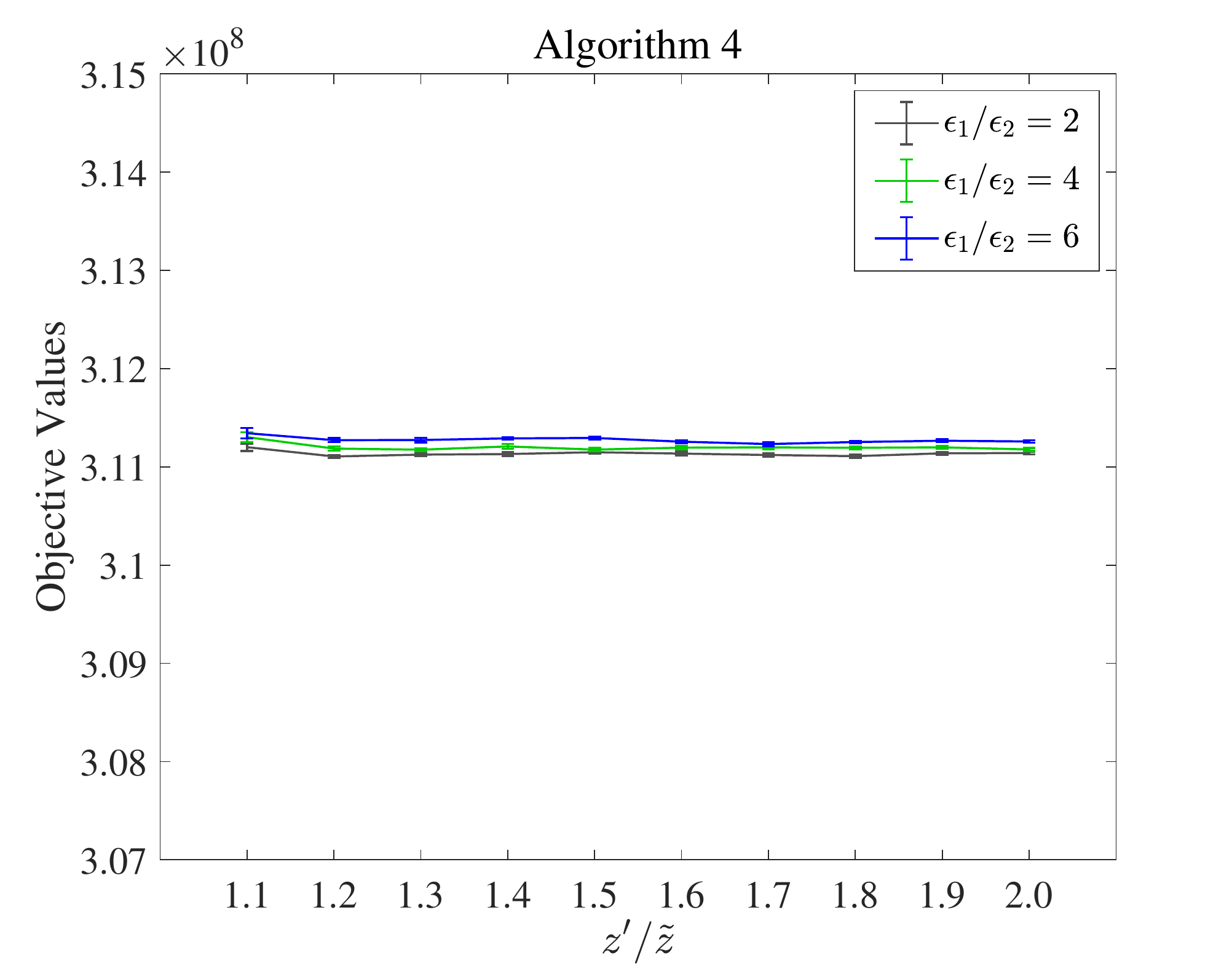}
		\includegraphics[width=0.24\textwidth]{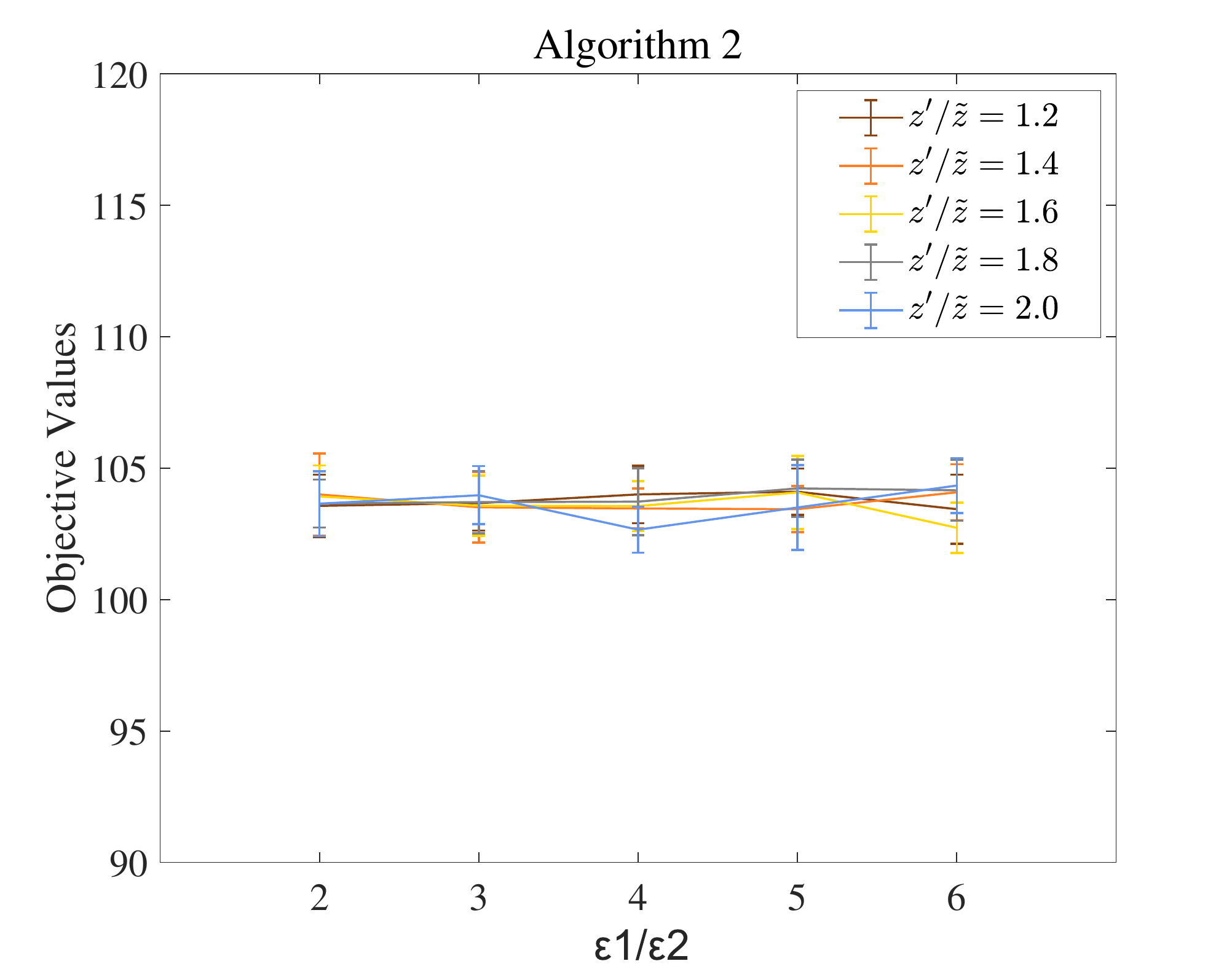}  
		\includegraphics[width=0.24\textwidth]{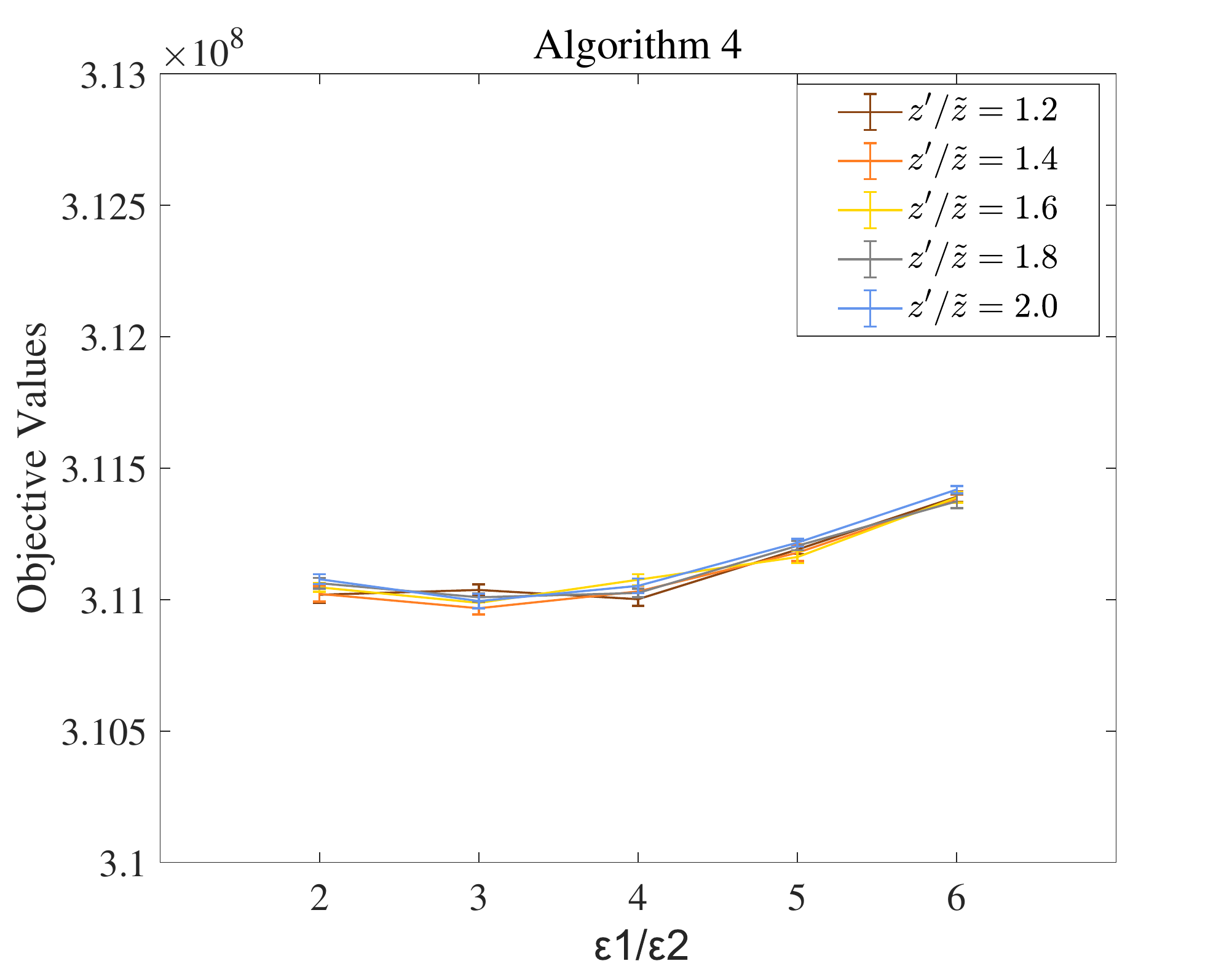} 
                                   \vspace{-0.05in}
		\caption{The  average objective values and standard deviations with varying $z'/\tilde{z}$ and $\epsilon_1/\epsilon_2$.}     
		\label{fig-stability}
	\end{center}
			 \vspace{-0.2in}
\end{figure}


\textbf{The influences of $z'$ and $\epsilon_1/\epsilon_2$ for Algorithm~\ref{alg-kc2} and~\ref{alg-km2}.} 
We study the influence of $z'$ on the performances. We vary $z'/\tilde{z}$ from $1.1$ to $2$, where $\tilde{z}=\frac{\epsilon_2}{k}|S|$ is the expected number of outliers in $S$. 
We also study the stability of the obtained result if we just run Algorithm~\ref{alg-kc2} ({\em resp.,} Algorithm~\ref{alg-km2}) by one time. We run each of them $200$ times on the synthetic datasets and show the results including the average objective values and standard deviations. 
 From Figure~\ref{fig-stability}, we can see that the performances are quite stable when varying $z'/\tilde{z}$ and $\epsilon_1/\epsilon_2$. 


\textbf{Precision and purity.}  To further evaluate our experimental results, we compute the measures {\em precision} and {\em purity}, which have been widely used before~\cite{DBLP:books/daglib/0021593}. The precision is the proportion of the ground-truth outliers found by the algorithm ({\em i.e.,} $\frac{|Out\cap Out_{truth}|}{|Out_{truth}|}$, where $Out$ is the set of returned outliers and $Out_{truth}$ is the set of ground-truth outliers). For each obtained cluster, we assign it to the ground-truth cluster which is most frequent in the obtained cluster, and the purity measures the accuracy of this assignment. 
%
%
%
%
Specifically, let $\{C_1, C_2, \cdots, C_k\}$ be the ground-truth clusters and $\{C'_1, C'_2, \cdots, C'_k\}$ be the obtained clusters from the algorithm; the purity is equal to $\frac{1}{n-z}\sum^k_{j=1}\max_{1\leq l\leq k}|C'_j\cap C_l|$. The experimental results suggest that our algorithms can achieve the precisions and the purities comparable to those of the baselines. Due to the space limit, we leave the details to Section~\ref{sec-precision}.

 \vspace{-0.1in}
\section{Future Work}
\vspace{-0.1in}

Following this work, an interesting question  is that whether the  significance criterion can be applied to analyze  the performance of uniform sampling for other well-known optimization problems, such as {\em PCA with outliers}~\cite{DBLP:journals/jacm/CandesLMW11} and {\em projective clustering with outliers}~\cite{feldman2011unified}. 



\newpage
\bibliographystyle{abbrv}

\bibliography{lightweight}

\section{Proof of Lemma~\ref{lem-imp1}}
\label{sec-proof-lem-imp1}
Lemma~\ref{lem-imp1} can be directly obtained through the following claim (we need to replace $\eta$ by $\eta/k$ in Claim~\ref{cla-sample}, for taking the union bound over all the $k$ clusters).

\begin{claim}
\label{cla-sample}
Let $U$ be a set of elements and $V\subseteq U$ with $\frac{|V|}{|U|}=\tau>0$. Given $\eta, \delta\in(0,1)$, one uniformly selects a set $S$ of elements from $U$ at random. (\rmnum{1}) If $|S|\geq \frac{1}{\tau}\log\frac{1}{\eta}$, with probability at least $1-\eta$, $S$ contains at least one element from $V$. (\rmnum{2}) If $|S|\geq\frac{3}{\delta^2\tau}\log\frac{2}{\eta}$, with probability at least $1-\eta$, we have $\big||S\cap V|-\tau |S|\big|\leq \delta\tau |S|$.
\end{claim}
\begin{proof}
Actually, (\rmnum{1}) is a folklore result having been presented in several papers before (such as~\cite{DX14}). Since each sampled element falls in $V$ with probability $\tau$, we know that the sample $S$ contains at least one element from $V$ with probability $1-(1-\tau)^{|S|}$. Therefore, if we want $1-(1-\tau)^{|S|}\geq 1-\eta$, $|S|$ should be at least $\frac{\log 1/\eta}{\log 1/(1-\tau)}\leq\frac{1}{\tau}\log\frac{1}{\eta}$.

(\rmnum{2}) can be proved by using the Chernoff bound~\cite{alon2004probabilistic}. Define $|S|$ random variables $\{y_1, \cdots, y_{|S|}\}$: for each $1\leq i\leq |S|$, $y_i=1$ if the $i$-th sampled element falls in $V$, otherwise, $y_i=0$. So $E[y_i]=\tau$ for each $y_i$. As a consequence, we have
\begin{eqnarray}
\textbf{Pr}(\big|\sum^{|S|}_{i=1}y_i-\tau |S|\big|\leq \delta\tau|S|)\geq 1-2e^{-\frac{\delta^2\tau}{3}|S|}. 
\end{eqnarray}
If $|S|\geq\frac{3}{\delta^2\tau}\log\frac{2}{\eta}$, with probability at least $1-\eta$, $\big|\sum^{|S|}_{i=1}y_i-\tau |S|\big|\leq \delta\tau|S|$ ({\em i.e.}, $\big||S\cap V|-\tau |S|\big|\leq \delta\tau |S|$).
\qed
\end{proof}

\section{Extensions of Theorem~\ref{the-km} and \ref{the-km2}}
\label{sec-extension} 
The results of Theorem~\ref{the-km} and \ref{the-km2} can be easily extended to $k$-median clustering with outliers in Euclidean space by using almost the same idea, where the only difference is that we can directly use triangle inequality in the proof ({\em e.g.,} the inequality (\ref{for-km-2-1}) is replaced by $||\tilde{q}-h_{\tilde{j}_q}||\leq ||\tilde{q}-q||+||q-h_{j_q}||$); the coefficients $\alpha$ and $\beta$ are reduced to be $\big(1+(1+c)\frac{1+\delta}{1-\delta}\big)$ and $(1+c)\frac{1+\delta}{1-\delta}$ ({\em resp.,} $\big(1+(1+c)\frac{t}{t-1}\frac{1+\delta}{1-\delta}\big)$ and $(1+c)\frac{t}{t-1}\frac{1+\delta}{1-\delta}$) in Theorem~\ref{the-km} ({\em resp.,} Theorem~\ref{the-km2}), respectively.


To solve the metric $k$-median/means clustering with outliers problems for an instance $(X, d)$, we should keep in mind that the cluster centers can only be selected from the vertices of $X$. However, the optimal cluster centers $O^*=\{o^*_1, \cdots, o^*_k\}$ may not be contained in the sample $S$, and thus we need to modify our analysis slightly. We observe that the sample $S$ contains a set $O'$ of vertices close to $O^*$ with certain probability. Specifically, for each $1\leq j\leq k$, there exists a vertex $o'_j\in O'$ such that $d(o'_j, o^*_j)\leq O(1)\times \frac{1}{|C^*_j|}\sum_{p\in C^*_j}d(p, o^*_j)$ (or $\big(d(o'_j, o^*_j)\big)^2\leq O(1)\times \frac{1}{|C^*_j|}\sum_{p\in C^*_j}\big(d(p, o^*_j)\big)^2$) with constant probability (this claim can be easily proved by using the Markov's inequality). Consequently, we can use $O'$ to replace $O^*$ in our analysis, and achieve the similar results as Theorem~\ref{the-km} and \ref{the-km2}.

\section{The Experimental Results for $K$-median Clustering with Outliers}
\label{sec-exp-kmed}

\begin{figure}
	\begin{center}
		\includegraphics[width=0.45\textwidth]{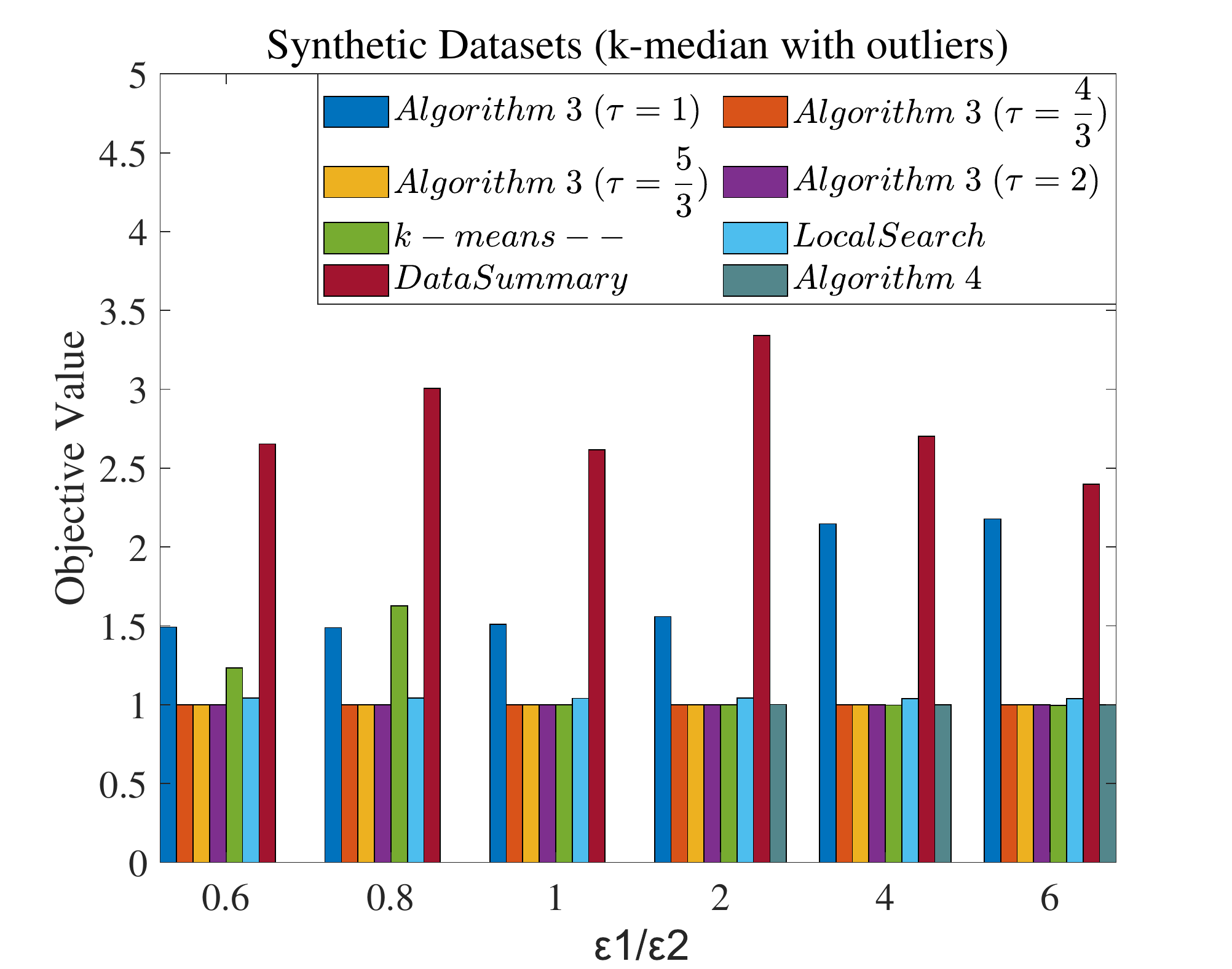} 
		\includegraphics[width=0.45\textwidth]{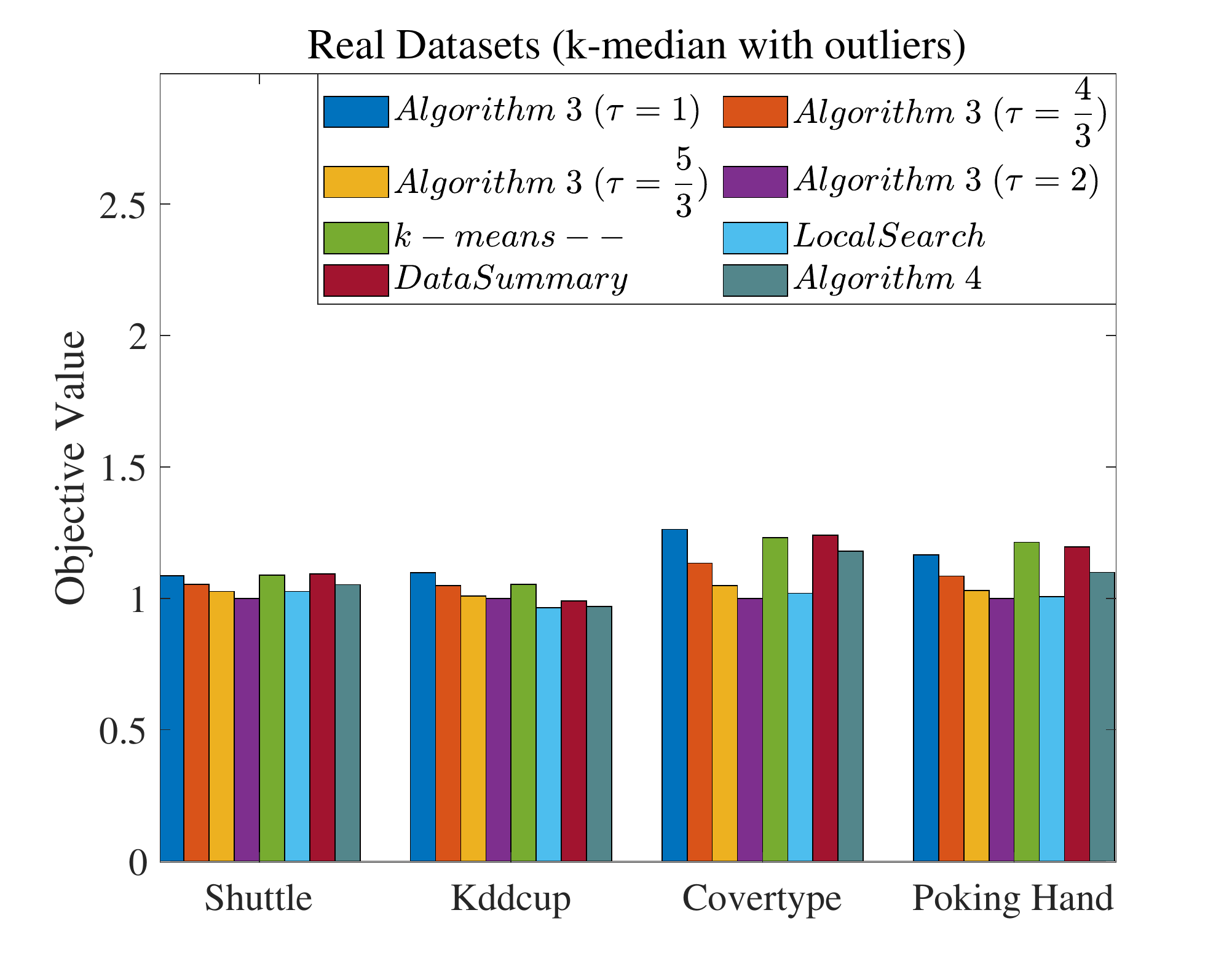} 
		\caption{The normalized objective values on the synthetic and real  datasets.}     
		\label{fig-objective}
	\end{center}
\end{figure}

\begin{figure}
	\begin{center}
		\includegraphics[width=0.45\textwidth]{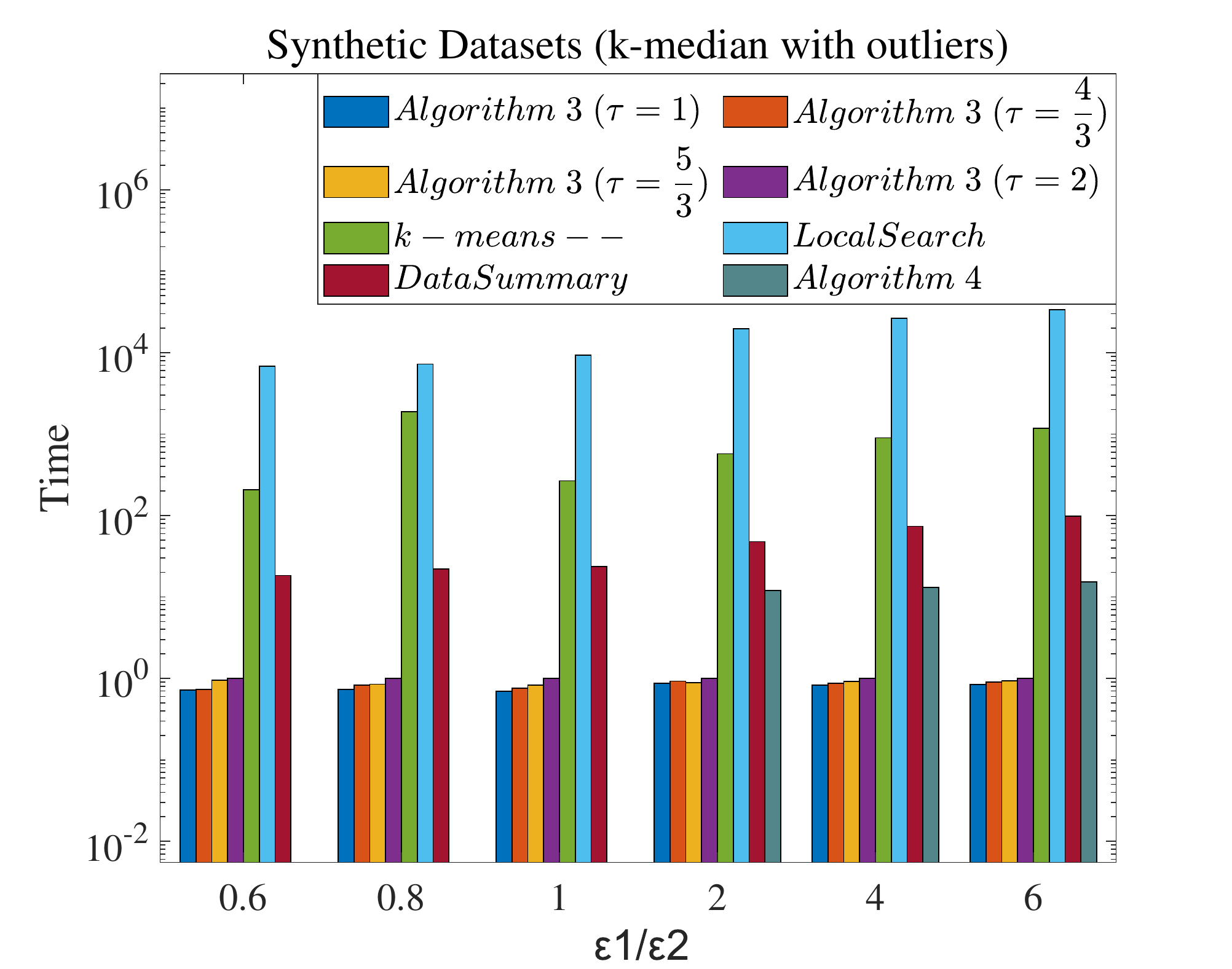} 
		\includegraphics[width=0.45\textwidth]{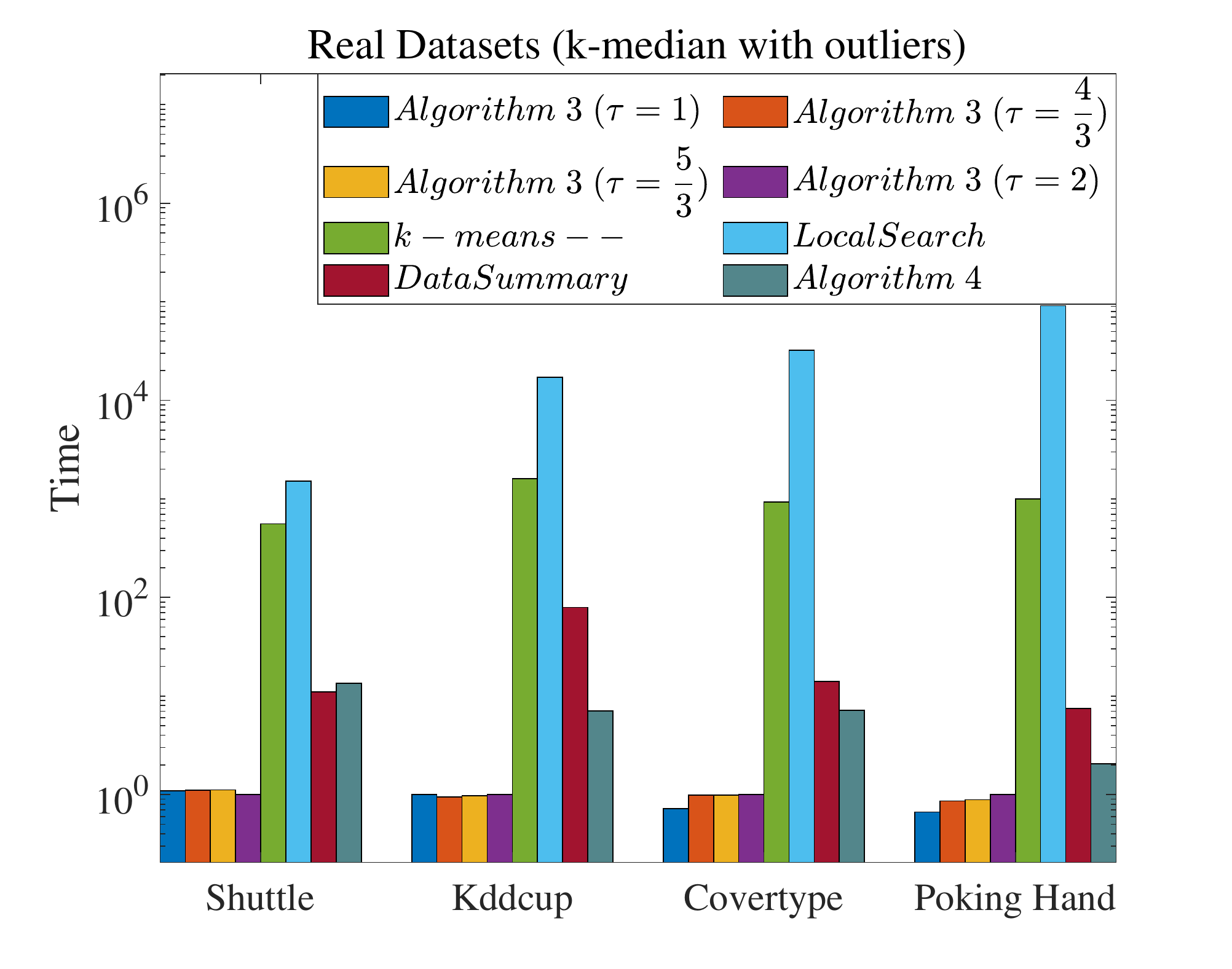}
		  \caption{The normalized running times on the synthetic and real  datasets.}     
		\label{fig-time}
	\end{center}
\end{figure}

\section{Precision and Purity}
\label{sec-precision}

\begin{table}[h]
	\centering
	\caption{Precision on synthetic datasets}
	\begin{tabular}{|c|c|c|c|c|c|c|}
		\hline
		$\epsilon_{1}/\epsilon_{2}=$ & 0.6   & 0.8   & 1     & 2     & 4     & 6 \\
		\hline
		Algorithm~\ref{alg-kc1} $\tau = 1$ & 0.512  & 0.734  & 0.755 & 0.746  & 0.774 & 0.876  \\
		\hline
		Algorithm~\ref{alg-kc1} $\tau = \frac{4}{3}$ & 0.808  & 0.897 & 0.839  & 0.999 & 0.999 & 0.852  \\
		\hline
		Algorithm~\ref{alg-kc1} $\tau = \frac{5}{3}$ & 0.874  & 0.845 & 0.973 & 0.999 & 0.999 & 1.000 \\
		\hline
		Algorithm~\ref{alg-kc1} $\tau = 2$ & 0.974  & 0.997 & 0.997 & 0.998 & 0.999 & 0.999 \\
		\hline
		Algorithm~\ref{alg-kc2}  & None   & None   & None   & 1.000 &1.000 & 1.000 \\
		\hline
		\textsc{MK} & 1.000 & 1.000 & 1.000 & 1.000& 1.000 & 1.000 \\
		\hline
		\textsc{Malkomes} & 1.000 & 1.000 &	1.000 & 1.000 & 1.000 & 1.000 \\
		\hline
		\textsc{Charikar} & 1.000 &1.000 & 1.000 & 1.000 & 1.000 & 1.000 \\
		\hline
		\textsc{DYW}    & 0.999 & 0.997 &	0.997 & 0.998 & 0.998&0.997 \\
	
		\hline
		Algorithm~\ref{alg-km} $\tau = 1$ & 1.000 & 1.000 &  1.000 & 1.000 & 1.000 & 1.000  \\
		\hline
		Algorithm~\ref{alg-km} $\tau = \frac{4}{3}$ & 1.000 & 1.000 & 1.000 &1.000 & 1.000& 0.999 \\
		\hline
		Algorithm~\ref{alg-km} $\tau = \frac{5}{3}$ &1.000 & 1.000 &1.000 & 1.000 &  0.999 &  0.999 \\
		\hline
		Algorithm~\ref{alg-km} $\tau = 2$&1.000  &1.000 & 0.999  & 0.999  & 0.997 & 0.998 \\
		\hline
		 Algorithm~\ref{alg-km2}   & None   & None   & None &  1.000 &  1.000 &  1.000 \\
		\hline
		\textsc{$k$-means$--$} & 1.000 &  1.000 &  1.000 &  1.000 &  1.000 &  1.000 \\
		\hline
		\textsc{LocalSearch} & 1.000 &  1.000 &  1.000 & 1.000& 1.000 &  1.000 \\
		\hline
		\textsc{DataSummary} & 1.000 & 0.999 & 1.000  & 1.000  & 1.000  & 1.000  \\

		\hline
		Algorithm~\ref{alg-km} ($k$-median) $\tau = 1$ & 0.994 & 0.998& 0.997  &0.998  &0.999 & 0.999 \\
	\hline
	Algorithm~\ref{alg-km} ($k$-median) $\tau = \frac{4}{3}$ & 1.000 & 1.000&  0.998 & 1.000 & 1.000 & 1.000 \\
	\hline
	Algorithm~\ref{alg-km} ($k$-median) $\tau = \frac{5}{3}$ & 1.000 & 1.000&  1.000 & 1.000 & 1.000 & 1.000\\
	\hline
	Algorithm~\ref{alg-km} ($k$-median) $\tau = 2$& 1.000 &1.000&  1.000 & 1.000 & 1.000 & 1.000\\
	\hline
	Algorithm~\ref{alg-km2} ($k$-median) & None   & None   & None &  1.000 &  1.000 &  1.000 \\
	\hline
	\textsc{$k$-means$--$} ($k$-median) & 1.000 &  1.000 &  1.000 &  1.000 &  1.000 &  1.000 \\
	\hline
	\textsc{LocalSearch} ($k$-median) & 1.000 &  1.000 &  1.000 & 1.000& 1.000 &  1.000 \\
	\hline
	\textsc{DataSummary} ($k$-median) & 0.958 & 0.932 & 0.948  & 0.927  & 0.964 & 0.938  \\
	\hline
	\end{tabular}%
	\label{tab-1}%
\end{table}%
\begin{table}[h]
	\centering
	\caption{Purity on synthetic datasets}
	\begin{tabular}{|c|c|c|c|c|c|c|}
	\hline
	$\epsilon_{1}/\epsilon_{2}=$ & 0.6   & 0.8   & 1     & 2     & 4     & 6\\
	\hline
	Algorithm~\ref{alg-kc1} $\tau = 1$ & 0.746 & 0.804 & 0.722&0.977   & 0.896 & 0.978\\
	\hline
	Algorithm~\ref{alg-kc1} $\tau = \frac{4}{3}$ & 0.860  & 0.981 & 0.931 &1.000 & 0.999  & 0.976  \\
	\hline
	Algorithm~\ref{alg-kc1} $\tau = \frac{5}{3}$ & 0.957   & 0.953  &0.893 &1.000 & 1.000 & 1.000  \\
	\hline
	Algorithm~\ref{alg-kc1} $\tau = 2$ & 0.972  & 1.000 & 1.000 &1.000  &1.000  &1.000  \\
	\hline
	Algorithm~\ref{alg-kc2}  & None   & None   & None   & 1.000 &1.000 & 1.000 \\
	\hline
	\textsc{MK} & 1.000 & 1.000 & 1.000 & 1.000& 1.000 & 1.000 \\
	\hline
	\textsc{Malkomes} & 1.000 & 1.000 &	1.000 & 1.000 & 1.000 & 1.000 \\
	\hline
	\textsc{Charikar} & 1.000 &1.000 & 1.000 & 1.000 & 1.000 & 1.000 \\
	\hline
	\textsc{DYW}   & 1.000 &1.000 & 1.000 & 1.000 & 1.000 & 1.000 \\
	
	\hline
	Algorithm~\ref{alg-km} $\tau = 1$ & 0.972 & 0.989& 0.986  &0.986  &0.961 & 0.914 \\
	\hline
	Algorithm~\ref{alg-km} $\tau = \frac{4}{3}$ & 1.000 & 1.000&  0.998 & 1.000 & 1.000 & 1.000 \\
	\hline
	Algorithm~\ref{alg-km} $\tau = \frac{5}{3}$ & 1.000 & 1.000&  1.000 & 1.000 & 1.000 & 1.000\\
	\hline
	Algorithm~\ref{alg-km} $\tau = 2$& 1.000 &1.000&  1.000 & 1.000 & 1.000 & 1.000\\
	\hline
	Algorithm~\ref{alg-km2}   & None   & None   & None &  1.000 &  1.000 &  1.000 \\
	\hline
	\textsc{$k$-means$--$} & 1.000 &  1.000 &  1.000 &  1.000 &  1.000 &  1.000 \\
	\hline
	\textsc{LocalSearch} & 1.000 &  1.000 &  1.000 & 1.000& 1.000 &  1.000 \\
	\hline
	\textsc{DataSummary} & 0.958 & 0.932 & 0.948  & 0.927  & 0.964 & 0.938  \\

	\hline
	Algorithm~\ref{alg-km} ($k$-median) $\tau = 1$ & 0.997 & 0.994& 0.999  &0.998  &0.996 & 0.994 \\
	\hline
	Algorithm~\ref{alg-km} ($k$-median) $\tau = \frac{4}{3}$ & 0.996 & 0.999&  1.000 & 1.000 & 1.000 & 1.000 \\
	\hline
	Algorithm~\ref{alg-km} ($k$-median) $\tau = \frac{5}{3}$ & 0.995 & 1.000&  1.000 & 1.000 & 1.000 & 1.000\\
	\hline
	Algorithm~\ref{alg-km} ($k$-median) $\tau = 2$& 1.000 &1.000&  1.000 & 1.000 & 1.000 & 1.000\\
	\hline
	Algorithm~\ref{alg-km2}  ($k$-median) & None   & None   & None &  1.000 &  1.000 &  1.000 \\
	\hline
	\textsc{$k$-means$--$} ($k$-median) & 0.996 &  0.999 &  1.000 &  1.000 &  1.000 &  1.000 \\
	\hline
	\textsc{LocalSearch} ($k$-median) &0.999 &  1.000 &  1.000 & 1.000& 1.000 &  1.000 \\
	\hline
	\textsc{DataSummary} ($k$-median) & 0.872 & 0.912 & 0.968 & 0.984  & 0.992 & 0.991  \\
	\hline
\end{tabular}%
	\label{tab-2}%
\end{table}%
\begin{table}[h]
	\centering
	\caption{Precision on real datasets}
	\begin{tabular}{|c|c|c|c|c|}
		\hline
		$datasets$ & Shuttle   & Kddcup   & Covtype     & Poking Hand     \\
		\hline
		Algorithm~\ref{alg-kc1} $\tau = 1$ & 0.904 & 0.966  & 0.612 &0.959   \\
		\hline
		Algorithm~\ref{alg-kc1} $\tau = \frac{4}{3}$ & 0.906 & 0.964  &0.695  &0.978  \\
		\hline
		Algorithm~\ref{alg-kc1} $\tau = \frac{5}{3}$&0.913   & 0.965  &0.753  &0.975   \\
		\hline
		Algorithm~\ref{alg-kc1} $\tau = 2$ &0.904& 0.966  & 0.908  &0.984  \\
		\hline
		Algorithm~\ref{alg-kc2}  		& 0.903   & 0.957 & 0.513  & 0.965 \\
		\hline
		\textsc{MK}    			& 0.903 & 0.978  &0.502  & 0.968\\
		\hline
		\textsc{Malkomes} 		& 0.933  & 0.959  &	0.754  & 0.957\\
		
		\hline
		\textsc{DYW}    				& 0.896 & 0.960  &	0.804  & 0.986  \\
		
		\hline
		Algorithm~\ref{alg-km} $\tau = 1$ &0.886  &0.946 & 0.823 &0.993\\
		\hline
		Algorithm~\ref{alg-km} $\tau = \frac{4}{3}$& 0.883 &0.947  &0.900 &0.991 \\
		\hline
		Algorithm~\ref{alg-km} $\tau = \frac{5}{3}$ &0.883  &0.947 &0.916  &0.990   \\
		\hline
		Algorithm~\ref{alg-km} $\tau = 2$&0.886  &0.948 &0.897   & 0.991  \\
		\hline
		Algorithm~\ref{alg-km2}  			& 0.906 & 0.958 & 0.807  &  0.986\\
		\hline
		\textsc{$k$-means$--$} 			& 0.883 & 0.971 & 0.793  &  0.999 \\
		\hline
		\textsc{LocalSearch}			 & 0.894 & 0.958  & 0.795& 0.973 \\
		\hline
		\textsc{DataSummary}			 & 0.889 & 0.950 & 0.764  & 0.988  \\

		\hline
		Algorithm~\ref{alg-km} ($k$-median) $\tau = 1$ &0.889  &0.946 & 0.728 &0.995\\
		\hline
		Algorithm~\ref{alg-km}  ($k$-median) $\tau = \frac{4}{3}$& 0.886 &0.946  &0.802 &0.990 \\
		\hline
		Algorithm~\ref{alg-km} ($k$-median) $\tau = \frac{5}{3}$ &0.881  &0.946 &0.802  &0.990   \\
		\hline
		Algorithm~\ref{alg-km} ($k$-median) $\tau = 2$&0.887  &0.949 &0.862   & 0.989  \\
		\hline
		Algorithm~\ref{alg-km2} ($k$-median)  			& 0.917 & 0.966 & 0.777  &  0.994\\
		\hline
		\textsc{$k$-means$--$}  ($k$-median)			& 0.882 & 0.966 & 0.733  &  0.999 \\
		\hline
		\textsc{LocalSearch} ($k$-median)			 & 0.891 & 0.953  & 0.751& 0.964 \\
		\hline
		\textsc{DataSummary}	 ($k$-median)		 & 0.882 & 0.966 & 0.734  & 0.999  \\
		\hline
	\end{tabular}%
	\label{tab-3}%
\end{table}%
\begin{table}[!htb]
	\centering
	\caption{Purity on real datasets}
	\begin{tabular}{|c|c|c|c|c|}
	\hline
	$datasets$ & Shuttle   & Kddcup   & Covtype     & Poking Hand     \\
	\hline
	Algorithm~\ref{alg-kc1} $\tau = 1$ &0.810 &0.690  &0.491  &0.502    \\
	\hline
	Algorithm~\ref{alg-kc1} $\tau = \frac{4}{3}$ &0.837 &0.848  & 0.493 &0.504  \\
	\hline
	Algorithm~\ref{alg-kc1} $\tau = \frac{5}{3}$ &0.859 & 0.872   &0.496  & 0.505  \\
	\hline
	Algorithm~\ref{alg-kc1} $\tau = 2$ & 0.913  &0.904  &0.504  & 0.505 \\
	\hline
	Algorithm~\ref{alg-kc2}  		& 0.830   & 0.851 & 0.495  & 0.504 \\
	\hline
	\textsc{MK}    			& 0.790  & 0.579   &0.513   & 0.501\\
	\hline
	\textsc{Malkomes} 		&0.789   & 0.632  &	0.498   & 0.500 \\
	
	\hline
	\textsc{DYW}    				& 0.790 & 0.580  &	0.492  & 0.501  \\
	
	\hline
	Algorithm~\ref{alg-km} $\tau = 1$ & 0.793   & 0.579 &  0.490  & 0.508\\
	\hline
	Algorithm~\ref{alg-km} $\tau = \frac{4}{3}$ & 0.790 &0.580   & 0.494  & 0.507  \\
	\hline
	Algorithm~\ref{alg-km} $\tau = \frac{5}{3}$ & 0.803  & 0.582  &0.504  &0.505  \\
	\hline
	Algorithm~\ref{alg-km} $\tau = 2$& 0.797  &0.582   & 0.508   &0.505  \\
	\hline
	Algorithm~\ref{alg-km2} 		& 0.793  & 0.579  &0.490 &  0.501 \\
	\hline
	\textsc{$k$-means$--$} 			& 0.818  & 0.579 &0.491  &  0.501 \\
	\hline
	\textsc{LocalSearch}			 & 0.790 & 0.579   &0.498 & 0.511 \\
	\hline
	\textsc{DataSummary}			 & 0.832 & 0.591  &0.488  & 0.504  \\

	\hline
	Algorithm~\ref{alg-km}  ($k$-median) $\tau = 1$ & 0.790   & 0.582 &  0.491  & 0.508\\
	\hline
	Algorithm~\ref{alg-km}  ($k$-median) $\tau = \frac{4}{3}$ & 0.798 &0.579   & 0.493  & 0.505  \\
	\hline
	Algorithm~\ref{alg-km}  ($k$-median) $\tau = \frac{5}{3}$ & 0.797  & 0.581  &0.499  &0.504  \\
	\hline
	Algorithm~\ref{alg-km}  ($k$-median) $\tau = 2$& 0.790  &0.582   & 0.508   &0.505  \\
	\hline
	Algorithm~\ref{alg-km2}  ($k$-median) 			& 0.789  & 0.579  &0.492 &  0.501 \\
	\hline
	\textsc{$k$-means$--$}  ($k$-median)			& 0.790  & 0.579 &0.492  &  0.501 \\
	\hline
	\textsc{LocalSearch}  ($k$-median)			 & 0.811 & 0.627   &0.496 & 0.510 \\
	\hline
	\textsc{DataSummary}  ($k$-median)			 & 0.829 & 0.602  &0.489  & 0.502  \\
	\hline
\end{tabular}%
	\label{tab-4}%
\end{table}%

\end{document}